\pdfoutput=1

\documentclass[11pt]{article}

\usepackage{amsmath,amssymb,amsthm}
\usepackage[margin=1.1in]{geometry}
\usepackage{graphicx}
\usepackage{booktabs}
\usepackage{array}
\usepackage{tabularx}
\usepackage{rotating}
\usepackage{xcolor}
\definecolor{linkblue}{RGB}{20,60,150}
\usepackage[square]{natbib}
\usepackage[colorlinks=true,linkcolor=linkblue,citecolor=linkblue,urlcolor=linkblue]{hyperref}

\theoremstyle{plain}
\newtheorem{theorem}{Theorem}[section]
\newtheorem{proposition}[theorem]{Proposition}

\newtheorem*{lemmafiveA}{Lemma 5.2A}
\theoremstyle{definition}
\newtheorem{definition}[theorem]{Definition}
\newtheorem{example}[theorem]{Example}
\newtheorem*{definitionstar}{Definition}

\newcommand{\cU}{\mathcal{U}}
\newcommand{\Msem}{M_{\mathrm{sem}}}
\newcommand{\Fsem}{F_{\mathrm{sem}}}
\newcommand{\Rrep}{R_{\mathrm{rep}}}
\newcommand{\Psem}{P_{\mathrm{sem}}}
\newcommand{\PE}{P_{E}}
\newcommand{\QE}{Q_{E}}
\newcommand{\ObsE}{O_{E}}
\newcommand{\IOb}{I_{\mathrm{Ob}}^{E}}
\newcommand{\rsem}{r_{\mathrm{sem}}}
\newcommand{\rE}{r_{E}}
\newcommand{\Ftwo}{\mathbb{F}_{2}}
\newcommand{\Zint}{\mathbb{Z}}
\newcommand{\At}{\mathrm{At}}
\newcommand{\ArchCtx}{\mathrm{ArchCtx}}
\newcommand{\code}[1]{\texttt{#1}}

\title{\textbf{SAGA: A Comparison Theorem for Local-to-Global Software Architecture}\\[0.6em]
\large\itshape From Semantic Repair Cohomology to Algebraic-Geometric Descent}
\author{Hiroyuki Nakahata\\
\normalsize Independent Researcher\\
\normalsize ORCID: \href{https://orcid.org/0009-0008-5928-0234}{0009-0008-5928-0234}}
\date{\normalsize Version 1.0.1 --- August 2026\\
\normalsize DOI: \href{https://doi.org/10.5281/zenodo.21603761}{10.5281/zenodo.21603761}
--- Licensed under \href{https://creativecommons.org/licenses/by/4.0/}{CC BY 4.0}}

\begin{document}

\maketitle

\begin{abstract}
In a software architecture, each individual service can obey its own
conventions and each handoff between adjacent services can hold, and yet a
semantic inconsistency may remain that appears only on a full traversal
of the system. This paper independently constructs two cohomologies
that measure this gap between local correctness and global correctness, and
proves that they agree. The first construction speaks the language of
\emph{repair}: from the semantic repair options admitted in each local
context and their equivalence relation, it generates the coefficient
$\Msem$. The second speaks the language of \emph{equations}: it organizes
the constraints of the architecture as a simultaneous equation system and
generates the quotient coefficient $\QE$ by its obstruction ideal. Over a
selected finite cover $\cU$ in Algebraic Architecture Theory (AAT) --- the
theory that constructs software architecture as algebraic geometry --- and
under finitely many selection conditions matching the local data, the
comparison map induces an isomorphism
$H^1_{\mathrm{sem}}(\cU)\cong\check{H}^1(\cU,\QE)$ together with a
correspondence of residual classes. We call this the SAGA comparison
theorem. The obstruction measured in the language of repair and the
obstruction measured in the language of equations are the same cohomology
class, so semantic diagnosis and geometric computation translate into each
other in both directions. Moreover, when the family of repair states
satisfies the sheaf condition, the existence of a global repair is
equivalent to the vanishing of the obstruction class on both sides.

The paper presents this result in three layers: the mathematical proof of
the comparison theorem (Sections~3--5); the Lean formalization status at
release time (Section~6); and a diagnosis in which the measurement tool
ArchSig, on a real open-source microservice system, reproducibly walks the
full circle from measuring a nonzero obstruction to its disappearance after
repair (Section~7). The three layers refer to the same release identity,
and each claim is connected to primary evidence. This yields, for
whole-loop inconsistencies that per-service verification cannot capture, a
connected diagnostic path of proved mathematics, machine-checked status,
and reproducible measurement.
\end{abstract}

\section{Introduction}\label{sec:1}

\subsection{The local-to-global architecture problem}\label{sec:11}

Each constituent of a software architecture is designed to be consistent
within its own local context. A service keeps its own data conventions, and
each individual handoff between modules satisfies its own contract.
Nevertheless, the system as a whole can exhibit semantic inconsistency. The
heart of this phenomenon is that the sum of local correctness does not
constitute global correctness.

This structure has the same shape as a classical problem type in
mathematics. Given a family of local data, where each local piece is
consistent and adjacent pieces are mutually consistent, does a global datum exist
that glues the whole together? Sheaf theory formalizes this question as the
sheaf condition, and cohomology measures the obstruction to gluing as a
class in $H^1$. This paper realizes this problem type on software
architecture and proves, as a comparison theorem, that the obstruction to
semantic repair and the obstruction of the geometry generated from the
architecture's simultaneous equation system coincide as the same
cohomology class.

\subsection{The SAGA research problem}\label{sec:12}

The research problem of SAGA (S\'emantique Architecturale, G\'eom\'etrie
Alg\'ebrique) is to set up the construction in the language of
\emph{repair} and the construction in the language of \emph{equations and
geometry} independently, and then to prove a comparison theorem connecting
their $H^1$ groups and residual classes (the two constructions are given in
Section~\ref{sec:4}). Once the comparison holds, semantic diagnoses
translate into geometric computations, and geometric computations translate
into semantic readings, in both directions.

Throughout this paper, the cohomology under discussion is the additive
\v{C}ech $H^1(\cU,-)$ relative to a selected monomorphic AAT cover $\cU$.
Identification with cover-independent sheaf cohomology is not among the
claims of this paper (\S\ref{sec:57}).

A remark on the name. The SAGA of this paper is unrelated to the Saga
pattern of compensating sequences known from distributed transactions
\citep{garciamolina1987sagas}. The name is an homage to Serre's GAGA
comparison theorem \citep{serre1956gaga}, which established the
correspondence between two geometries (\S\ref{sec:92}).

\subsection{The three layers of the result and the organization of the paper}\label{sec:13}

This paper presents the result in three layers. The mathematical layer is
the principal contribution of the paper: we construct the comparison
between semantic repair cohomology and equation-generated AAT \v{C}ech
cohomology, and prove the $H^1$ isomorphism, the residual class
correspondence, and the equivalence with global repair under the true
sheaf condition (Section~\ref{sec:5}). The Lean layer is the status at
release time: it records the formalized definitions, theorems, witnesses,
and proof chain at declaration granularity, fixing the correspondence with
the theorems of this paper and the axiom status (Section~\ref{sec:6}). The
ArchSig layer is a finite case study: ArchSig is a Rust command-line tool
that computes diagnoses from observations of the code and a selected
equation system; on a fixed microservice architecture, it shows in a
reproducible form how the mathematical objects are computed from finite
architecture evidence and how they change before and after repair
(Section~\ref{sec:7}).

In this paper, the \textbf{release identity} is the triple consisting of
the release tag of the Lean sources, the version of ArchSig, and the
digests of the input artifacts. The Lean sources and the ArchSig sources
are available in the public repository
\begin{center}
\href{https://github.com/iroha1203/AlgebraicArchitectureTheoryV2}{\nolinkurl{github.com/iroha1203/AlgebraicArchitectureTheoryV2}}.
\end{center}
The formalization status of Section~\ref{sec:6} and the reproduction
procedure of Section~\ref{sec:7} are stated against this fixed identifier.

The rest of the paper is organized as follows. Section~\ref{sec:2} presents
the AAT approach; Section~\ref{sec:3} the mathematical foundations needed
for SAGA; Section~\ref{sec:4} the independent construction of the two
complexes; Section~\ref{sec:5} the comparison theorem; Section~\ref{sec:6}
the Lean status; Section~\ref{sec:7} the ArchSig diagnosis of real code;
Section~\ref{sec:8} related work; Section~\ref{sec:9} the research
outlook; and Section~\ref{sec:10} the conclusion.

\section{The AAT Approach}\label{sec:2}

Algebraic Architecture Theory (AAT) is a mathematical theory that
constructs software architecture as algebraic geometry. This section
briefly describes what is distinctive about the approach and the skeleton
of the construction. All definitions of the mathematical objects needed
for the statement and proof of SAGA are given in Section~\ref{sec:3}.

The construction proceeds in three stages. First, AAT axiomatizes
primitive architectural facts as \textbf{Atoms}. An Atom does not belong
to any particular programming language, framework, or serialization
format, so facts spanning multiple languages, services, and storage
representations can be treated on the same architecture object. To speak
about inconsistencies where the implementation types agree but the
semantic conventions differ, the architecture must be constructed at a
level of abstraction above the type system, and this language independence
is the prerequisite for the subject of this paper. Second, multiple
architecture conditions are organized as an \textbf{Atom-indexed
architectural equation system} $E$. From $E$ one constructs residuals,
the obstruction ideal, and the locus of points where the required
equations hold simultaneously; the place where the constraints are
satisfied becomes a geometric object, the analogue of the zero set of a
system of equations. Third, local contexts, covers, restrictions, and
overlaps are organized as an \textbf{AAT site}. The local-to-global
problem thereby becomes expressible in the language of sheaves and
cohomology: whether a family of local conditions glues globally is a
question about the existence of a global section, and the obstruction to
gluing is a class in \v{C}ech cohomology.

Under this construction, architectural failures are structured as
residuals, ideals, and cohomology classes, and diagnosis and repair can be
read as changes of the same geometric objects. A repair is an operation
that kills a class, and the success of a repair is judged by the vanishing
of the class. Furthermore, ringed geometry over an architecture opens
connections to the tools of algebraic geometry --- schemes, derived
intersections, deformations, monodromy --- but what SAGA uses in this
paper is \v{C}ech $H^1$ and descent; the deeper tools are discussed in the
research outlook of Section~\ref{sec:9}. The algebraic-geometric character
of SAGA lies on the side of the construction: the observable ring
$\ObsE$, the witness ideals and the obstruction ideal, the quotient
coefficient $\QE=\ObsE/\IOb$, the AAT site, and the descent
interpretation that reads the zero class as a global repair
(Sections~\ref{sec:3}--\ref{sec:5}; the attribution of what is proper to
SAGA is in \S\ref{sec:51}).

The geometry of AAT is a \textbf{relative geometry} with the vocabulary,
the equation system, the coverage, and the coefficients held fixed. Every
claim is relative to the selected input data, and the provenance from
inputs to conclusions can be traced. This relativity is a constraint, but
it is also the source of the discipline that matches claims with
evidence; the measurement of Section~\ref{sec:7} is its executable form.

SAGA concretizes the local-to-global capability of this approach by
comparing the semantic repair obstruction with the equation-generated
\v{C}ech obstruction. The approach of this section is summarized in one
sentence:

\begin{quote}
AAT constructs software architecture as a relative algebraic geometry
generated from primitive architectural facts and simultaneous
architectural equations, so that local consistency, global obstruction,
and repair become geometric objects.
\end{quote}

\section{AAT Foundations for SAGA}\label{sec:3}

This section defines, in a form self-contained within the paper, the
mathematical objects needed to read the statement and proof of SAGA. The
arguments of the following sections presuppose only the definitions of
this section.

\subsection{Atoms and Atom families}\label{sec:31}

\begin{definition}[Atom]\label{def:atom}
An \emph{Atom} is a typed fact of the form
\begin{verbatim}
a = (kind, axis, subject, predicate, payload)
\end{verbatim}
where \code{kind} is the sort of the Atom (\code{component},
\code{relation}, \code{state}, \code{contract}, \code{semantic}, and so
on), \code{axis} is the signature / equation axis under which the Atom is
read, \code{subject} is the subject of the statement, \code{predicate} is
a single atomic statement made about the \code{subject}, and
\code{payload} is the content data of the statement. An Atom is a
primitive architectural fact belonging to no particular language or
framework, and inside AAT it is treated as a generator. The schema of the
core family includes \code{component(c)}, \code{relation\_r(c,d)},
\code{state(c,x)}, \code{contract(m,p)}, \code{semantic(t,s)},
\code{runtime\_interaction(u,v,h)}, among others. For instance,
\code{semantic(t,s)} is the semantic Atom whose subject is the
representation \code{t}, whose predicate is ``\code{t} obeys the semantic
convention \code{s}'', and whose payload is the description of \code{s}.
\end{definition}

We give one correspondence with real code. Suppose the service that
issues orders assembles an order number as
\begin{verbatim}
orderId = "ORD-" + zeroPad(seq.next(), 8)
\end{verbatim}
Then from this single site the following Atoms can be read:
\begin{verbatim}
component(orders)        -- the service this code belongs to
state(orders, orderId)   -- the local state "issued order numbers"
semantic(orderId, s)     -- s = the notation convention
                         --     "prefix ORD- + 8-digit zero-padded sequence"
\end{verbatim}
Observing the implementation type (\code{string}) alone does not yield
the last semantic Atom; reading the expression and how it is used is
required. Where Atoms come from --- extraction from real code --- belongs
to the observation process outside AAT, and its reproducible fixing is
described in Appendix~\ref{app:C2}. Inside AAT, Atoms are the given generators.

\begin{definition}[Atom family and support]\label{def:atomfamily}
Write $\At$ for the Atom universe. An \emph{Atom family} $F$ is a subset
of $\At$, and $\operatorname{support}(F)$ is the set of subjects occurring
in $F$. The family $F$ may contain component, relation, state, contract,
effect, and semantic Atoms in combination. This mixing is not excluded; the
question is which equations the mixed primitive facts satisfy and which
obstructions they generate.
\end{definition}

\subsection{Architecture objects}\label{sec:32}

\begin{definition}[architecture object]\label{def:archobject}
An \emph{architecture object} is a tuple
\[
A = (C, S, Q)
\]
obtained by equipping an Atom configuration $C=(F,\mathit{Rel},\mathit{Pl})$
(an Atom family $F$, relational data $\mathit{Rel}$, and placement
$\mathit{Pl}$) with structure maps $S$ (structural maps to graphs,
categories, algebras, state transitions, and so on) and selected
quantities $Q$ (invariants, measures, signature axes). We write
$F \Rightarrow A$ when $A$ is obtained from the Atom family $F$. Every
architecture object is generated from an Atom family (the
\emph{Atom-origin} principle). In this paper the fixed architecture
object is written $X$.
\end{definition}

\subsection{Architecture contexts and the context category}\label{sec:33}

\begin{definition}[architecture context]\label{def:archcontext}
An \emph{architecture context} $W$ for an architecture object is a local
reading of it. In the minimal model it is read as
\[
W = (\operatorname{Supp}(W), \operatorname{Ax}(W), \operatorname{Obs}(W)),
\]
where $\operatorname{Supp}(W)$ is the Atom support read in $W$,
$\operatorname{Ax}(W)$ is the selected signature / equation axis, and
$\operatorname{Obs}(W)$ is the family of coordinates, witnesses, and
semantic data readable in $W$. Component-local views, feature views,
semantic contract slices, and runtime trace slices are representative
contexts.
\end{definition}

\begin{definition}[context category]\label{def:contextcategory}
Write $\ArchCtx(X)$ for the category whose objects are contexts and whose
morphisms are reinterpretations $g:W'\to W$ (restriction, projection,
refinement, embedding, base change). Morphisms point from local contexts
to global contexts. The overlap of a cover is constructed as the pullback
\[
U_i\times_W U_j
\]
(the context overlap).
\end{definition}

\paragraph{Minimal poset model.}
When supports, axes, and observable families have finite meets, and under
the equivalence quotient in which there is at most one reinterpretation
morphism between any two contexts, $\ArchCtx_{\min}(X)$ becomes a
finite-meet poset category, and overlaps are computed as meets. In this
model every morphism is a monomorphism (\S\ref{sec:37} connects this fact
to the cover assumption). The finite computations of Section~\ref{sec:7}
are the executable form of this regime; the range of applicability of the
theorem over general $\ArchCtx(X)$ and the status of the formalization
are summarized in \S\ref{sec:57}.

\subsection{Atom-indexed architectural equation systems}\label{sec:34}

\begin{definition}[equation system]\label{def:equationsystem}
Fix the Atom universe $\At$, a family of architecture objects, and a
small category of local contexts. An \emph{Atom-indexed architectural
equation system} $E$ consists of the following data:
\begin{itemize}
\item $K_E$: the type of equation indices;
\item $\mathrm{role}_E : K_E \to \{\mathrm{required}, \mathrm{optional},
      \mathrm{derived}\}$: the role of each equation;
\item $\ObsE : \ArchCtx(X)^{\mathrm{op}} \to \mathbf{CommRing}$: a
      presheaf assigning to each context $W$ an observable ring
      $\ObsE(W)$;
\item $\nu_{W,i,a} \in \ObsE(W)$: the symbolic violation coordinate for
      $i \in K_E$ and $a \in \At$;
\item $\epsilon_{W,A,i,a} \in \ObsE(W)$: the object-dependent equation
      residual, evaluating equation $i$ on the architecture object $A$.
\end{itemize}
\end{definition}

The two coordinate families commute with restriction. The $\nu$ are the
symbolic coordinates generating the witness ideals, and the $\epsilon$
are the residual coordinates judging equation fulfillment. An equation
index $i$ \emph{holds} on $A$ when $\epsilon_{W,A,i,a}=0$ for all
contexts and Atoms. The simultaneous fulfillment of the required
equations is the basic condition for the consistency of the architecture.

Natural-language expressions of design principles and conventions are
derived from $E$ as readings of equation indices and their coordinate
families; they carry no independent truth predicate. The primary
mathematical object is $E$.

\subsection{AAT sites, presheaves, and the sheaf condition}\label{sec:35}

\begin{definition}[AAT site]\label{def:aatsite}
Fix an architecture object $X$, an equation system $E$, a signature
$\mathit{Sig}$ (the family of selected signature axes; in this paper it
is treated as part of the fixed package and not expanded further),
coverage requirements $R$, and a context-overlap package $\mathit{Ov}$.
A coverage family $\{W_i\to W\}$ is \emph{$(E,R,\mathit{Ov})$-admissible}
when all the data that $R$ requires to be readable on $W$ --- the support
of the target Atoms, the support of the coordinates of $E$ ($\nu$,
$\epsilon$), the selected witnesses and signature axes, and the
interaction overlaps between contexts --- are readable across the cover
through restriction to the $\{W_i\}$. Write $J_{E,R,\mathit{Ov}}$ for the
Grothendieck topology generated by the admissible families. The
\emph{AAT site} is the package
\[
\mathrm{Site}_{\mathrm{AAT}}(X,E,\mathit{Sig},R,\mathit{Ov})
=(\ArchCtx(X), E, \mathit{Sig}, R, \mathit{Ov}, J_{E,R,\mathit{Ov}}).
\]
In this paper we abbreviate the underlying site with this package fixed
as $S_X$, and write $W$ for its contexts. Covers do not generate Atoms
(non-generation). What the proofs of this paper consume from this
topology is only that $\cU$ belongs to the topology and the sheaf
condition (Theorem~\ref{thm:gluing}); the internal structure of
admissibility is used in no proof. The regime of finite execution is the
finite-meet poset model of \S\ref{sec:33}.
\end{definition}

\begin{definition}[presheaf and sheaf condition]\label{def:presheaf}\sloppy
A presheaf on $S_X$ is a functor $F:\ArchCtx(X)^{\mathrm{op}}\to
\mathbf{Set}$. $F$ is a \emph{sheaf} when, for every cover
$\{W_i\to W\}$, compatible local sections agreeing on overlaps glue to a
unique global section. The true semantic repair sheaf condition of
Section~\ref{sec:5} is an instance of this sheaf condition.
\end{definition}

\subsection{Witness ideals, the obstruction ideal, and the equation-generated coefficient}\label{sec:36}

\begin{definition}[witness ideal and obstruction ideal]\label{def:witnessideal}
The local witness ideal of an equation index $i$ is the ideal generated
by the symbolic violation coordinates,
\[
I_i^E(W)=\langle \nu_{W,i,a}\mid a\in \At\rangle\subset \ObsE(W).
\]
The sum of the witness ideals of the required equations,
\[
\IOb(W)=\sum_{i\in K_E,\ \mathrm{role}_E(i)=\mathrm{required}} I_i^E(W),
\]
is called the \emph{local obstruction ideal}. From the restriction
compatibility of $\nu$ it follows that, for every context morphism
$g:W'\to W$, $\mathrm{res}_g(\IOb(W))\subset \IOb(W')$, so the
context-wise ideals form an ideal subpresheaf.
\end{definition}

\begin{theorem}[generated obstruction quotient]\label{thm:quotient}
Fix an equation system $E$ and a displayed equation source (a choice, for
finitely many local contexts $W_q\to W_{\mathrm{base}}$, of an
architecture object $A_q$, a required equation index $i_q$, and a support Atom
$a_q$; in the cover-indexed case the local contexts are taken at the
charts). Then the following hold.
\begin{enumerate}
\item \textbf{generated coefficient}: the quotient
\[
\QE(W)=\ObsE(W)/\IOb(W)
\]
carries a restriction induced by the universal property of quotients,
and $W\mapsto \QE(W)$ is an abelian-group-valued presheaf.
\item \textbf{generated interpretation}: the interpretation of a
displayed source $q$ is constructed as the quotient class of the
residual, $\mathrm{interpret}(q)=[d_q]\in \QE(W_q)$ with
$d_q:=\epsilon_{W_q,A_q,i_q,a_q}$. The interpretation is not free data;
it is constructed.
\item \textbf{residual restriction naturality}: for a context morphism
$g:Z\to W_q$, setting $d_{q|Z}:=\epsilon_{Z,A_q,i_q,a_q}$,
\[
\mathrm{res}_g([d_q])=[\mathrm{res}_g(d_q)]=[d_{q|Z}].
\]
The restriction of a residual class is the class of the residual
evaluated with the same architecture reading, equation index, and
support Atom.
\item \textbf{vanishing}: if the displayed equations are fulfilled, then
$[d_q]=0$ for every $q$.
\item \textbf{quotient zero criterion}: $[d_q]=0 \iff d_q\in\IOb(W_q)$.
Hence if $d_q\notin\IOb(W_q)$ then $[d_q]\neq0$.
\end{enumerate}
\end{theorem}

\begin{proof}
By the general theory of quotients and ideals. (1) is the ideal
subpresheaf property of Definition~\ref{def:witnessideal} and the
universal property of quotients; (2) is a construction; (3) follows from
the restriction compatibility of $\epsilon$
(\S\ref{sec:34}) and computation with
representatives in the quotient; (4) is the definition of equation
fulfillment ($\epsilon=0$); (5) is the very definition of zero in a
quotient.
\end{proof}

What clause~5 provides is the zero test in the quotient. To pass from a
displayed failure to a nonzero class $[d_q]\neq0$, one separately needs
the semantic failure to materialize as a residual not belonging to the
ideal (\textbf{semantic faithfulness}). This is not a property of the
quotient coefficient but a condition belonging to the selection and
supply of the displayed source. Its bearer in actual measurement is
described in Section~\ref{sec:7} (\S\ref{sec:74}).

This $\QE$ is the coefficient of the geometric \v{C}ech complex of SAGA.
Intuitively, $\QE$ views the observables coarsely, exactly up to the
obstruction ideal: two observables have the same class precisely when
their difference is a finite $\ObsE$-combination of violation
coordinates. That the obstruction is ideal-theoretic --- that failure is
measured not as a label but as a class in a quotient --- is the
foundation of every construction from Section~\ref{sec:4} onward.

\subsection{Monomorphic AAT covers}\label{sec:37}

Fix a finite index set $I$ and an AAT cover
\[
\mathcal U=\{u_i:U_i\to W\}_{i\in I},
\]
and assume each $u_i$ is a monomorphism. Such a cover is called a
\textbf{monomorphic AAT cover}. Choose a total order on the indices and
write, for $i<j<k$,
\[
U_{ij}=U_i\times_W U_j,\qquad
U_{ijk}=U_i\times_W U_j\times_W U_k.
\]
The charts $U_i$ together with these pairwise / triple intersections are
called the \textbf{cover intersection diagram}. The three pairwise
intersections of a nonempty $U_{ijk}$ are nonempty, since they receive
projections from $U_{ijk}$. Empty pullbacks are excluded from the
objects of the intersection diagram, and their values are fixed by the
empty-overlap normalization (input~8 of Theorem~\ref{thm:central} in
Section~\ref{sec:5}). Finiteness is not needed for the comparison core
itself, but a finite cover is fixed so that the central theorem, the
finite witnesses, and the executable realization can be treated in the
same notation. In the finite-meet poset model of \S\ref{sec:33} every
morphism is a monomorphism, so the assumption is automatic; over general
$\ArchCtx(X)$ it is an explicit hypothesis of the theorem.

\subsection{Selected, generated, proved}\label{sec:38}

In SAGA, inputs, constructions, and consequences are strictly
distinguished.

\begin{center}
\small
\begin{tabular}{@{}p{0.17\linewidth}p{0.76\linewidth}@{}}
\toprule
Kind & Content \\
\midrule
selected &
the stage of the comparison (the monomorphic AAT cover, the
empty-overlap normalization); the primary data of the semantic side (the
semantic atom system, the local repair relation, the local repair
atlas); the primary data of the equation side (the equation system $E$,
the local equation-lift atlas); the correspondence between the two
(Atom/equation interpretation, the local-state interpretation $\beta$);
the two completeness conditions; the true sheaf condition used for
global gluing. Fixed as the cover, inputs 1--8, and the sheaf
hypothesis of stage (iii) of Theorem~\ref{thm:central} \\
\addlinespace
generated &
$\Msem$, $\QE$, the two \v{C}ech complexes, $\rsem$, $\rE$, the
coefficient map $\Phi$, the cochain map $\kappa$ \\
\addlinespace
proved &
repair-relation soundness, SAGA presentation exactness, the isomorphism
property of $\Phi$, commutation with the differentials, preservation of cocycles /
coboundaries, the $H^1$ isomorphism, the residual class correspondence,
the equivalence of the zero class with an actual global repair \\
\bottomrule
\end{tabular}
\end{center}

This distinction fixes the provenance of the conclusions while
preserving mathematical relativity: \emph{selected} is the input
contract, and \emph{proved} is the theorem under that contract.

\section{Semantic Repair and Equation-Generated Geometry}\label{sec:4}

This section constructs the two complexes of SAGA independently. The
semantic side is generated from semantic atoms and a repair relation, the
equation side from an equation system; each produces its coefficient,
complex, and residual from its own primary data alone. Neither
construction refers to the other, and the comparison of
Section~\ref{sec:5} --- the existence of the map, the fact that it is an
isomorphism, the correspondence of residual classes --- is the conclusion of a theorem,
not a repetition of the constructions. The single map connecting the two,
$\chi^E$ (Proposition~\ref{prop:chiE}), is a construction of a comparison
input for Section~\ref{sec:5} and is used in the construction of neither
complex.

\subsection{The cover-relative \v{C}ech complex}\label{sec:41}

Let $F$ be an abelian-group-valued presheaf on $S_X$. Define the
three-term cochain complex relative to $\cU$ by
\[
C^0(\mathcal U,F)=\prod_iF(U_i),\qquad
C^1(\mathcal U,F)=\prod_{i<j}F(U_{ij}),\qquad
C^2(\mathcal U,F)=\prod_{i<j<k}F(U_{ijk})
\]
(the products run over nonempty intersections), with the differentials on
ordered indices
\[
(\delta^0a)_{ij}=a_j|_{U_{ij}}-a_i|_{U_{ij}},
\qquad
(\delta^1c)_{ijk}=c_{jk}|_{U_{ijk}}-c_{ik}|_{U_{ijk}}+c_{ij}|_{U_{ijk}}.
\]
A direct computation gives $\delta^1\delta^0=0$, so the cover-relative
$H^1$ is defined by
\[
\check H^1(\mathcal U,F)=\ker\delta^1/\operatorname{im}\delta^0.
\]

\subsection{The semantic-side construction}\label{sec:42}

\paragraph{Semantic repair presentation.}
To each context $V$ of $S_X$, assign the set $\Lambda(V)$ of semantic
atoms distinguished on $V$, a projection $\pi_V:\Lambda(V)\to\At$ giving
the underlying Atom of each semantic atom, and a subset
$S(V)\subseteq\Lambda(V)$ of \textbf{supported semantic atoms} available
for repair. For each context morphism $V'\to V$ there is a functorial
restriction map $\Lambda(V)\to\Lambda(V')$ that preserves support (if
$\lambda\in S(V)$ then $\lambda|_{V'}\in S(V')$), and the projection
preserves the underlying Atom:
\[
\pi_{V'}(\lambda|_{V'})=\pi_V(\lambda).
\]
Write $\Fsem(V)$ for the free abelian group on $S(V)$; its elements
(formal finite sums of supported atoms) are called \textbf{repair
words}. On each $V$, select a restriction-stable subgroup
\[
\Rrep(V)\subset \Fsem(V)
\]
(the \textbf{local repair relation}). This is all of the primary data of
the semantic side. The comparison core
(\S\ref{sec:53}--\S\ref{sec:55}) uses these data restricted to the cover
intersection diagram.

\paragraph{Semantic coefficient.}
The quotient
\[
\Msem(V)=\Fsem(V)/\Rrep(V)
\]
carries a restriction induced by restriction-stability, and
$V\mapsto \Msem(V)$ is an abelian-group-valued presheaf on $S_X$. The
complex of \S\ref{sec:41} with coefficients $\Msem$,
$C^\bullet_{\mathrm{sem}}(\cU):=C^\bullet(\cU,\Msem)$, is the semantic
complex.

\paragraph{Affine semantic repair system.}
Let $\Psem$ be a presheaf of semantic local repair states on $S_X$. On
each $V$, $\Fsem(V)$ acts on $\Psem(V)$, restriction commutes with the
action, and on each $V$ of the cover intersection diagram the following
three conditions hold:
\begin{itemize}
\item \textbf{action soundness}: words in $\Rrep(V)$ act as the
identity;
\item \textbf{stabilizer completeness}: if $\Psem(V)$ is nonempty, every
word acting as the identity belongs to $\Rrep(V)$;
\item \textbf{local transitivity}: any two states are carried to each
other by the action of some word.
\end{itemize}
By soundness the action descends to an action of $\Msem(V)$; by
completeness this action is free; by transitivity it is transitive on
nonempty $\Psem(V)$. Therefore, on each $V$ of the cover intersection
diagram, a nonempty $\Psem(V)$ is an $\Msem(V)$-torsor (an affine
space). A torsor is a space without a distinguished origin in which,
instead, the difference of any two states is measured uniquely as an
element of the coefficient group. The construction of residuals below
uses only this difference, never the states themselves.

\paragraph{Semantic residual.}
Choosing a selected local repair atlas $\{p_i\}$, the torsor difference
on each nonempty overlap determines a unique
\[
r_{\mathrm{sem},ij}=p_j|_{U_{ij}}-p_i|_{U_{ij}}\in \Msem(U_{ij}),
\]
giving the \textbf{semantic residual} $\rsem\in C^1_{\mathrm{sem}}(\cU)$.
The residual $\rsem$ is a cocycle: on $U_{ijk}$, abbreviating
$r_{ij}=r_{\mathrm{sem},ij}$, we have
$p_k=r_{jk}+p_j=r_{jk}+r_{ij}+p_i$ and $p_k=r_{ik}+p_i$, and freeness of
the action gives $r_{jk}-r_{ik}+r_{ij}=0$.

The class $[\rsem]\in
H^1_{\mathrm{sem}}(\cU):=\check H^1(\cU,\Msem)$ does not depend on the
choice of atlas (\textbf{choice independence}): for another atlas
$p'_i=a_i+p_i$ the same torsor computation gives
$r'_{\mathrm{sem}}=\rsem+\delta^0a$. Consequently, the vanishing
$[\rsem]=0$ is equivalent to the existence of a correction
$a\in C^0_{\mathrm{sem}}(\cU)$ such that the corrected atlas
$(-a_i)+p_i$ agrees on all pairwise overlaps (the existence of a
\textbf{matching correction}). This paraphrase provides the intuition
for coboundaries: a coboundary $\delta^0a$ is the difference produced by
merely changing the reference (the atlas) within each chart. A zero
class is an inconsistency that can be removed by per-chart changes of
reference; a nonzero class is an inconsistency that no combination of
local reference changes removes, and that remains on a closed loop.

\subsection{The equation-side construction}\label{sec:43}

The equation side constructs the \textbf{geometric \v{C}ech complex}
$C^\bullet_E(\cU):=C^\bullet(\cU,\QE)$ with coefficients in the $\QE$ of
Theorem~\ref{thm:quotient}. The \textbf{equation-lift system} $\PE$ is a
presheaf on $S_X$ such that on each intersection $V$, the nonempty
$\PE(V)$ carries a free and transitive action of $\QE(V)$, and
restriction commutes with the action. From a selected local lift atlas
$\{e_i\}$, the difference on overlaps generates the geometric-side
residual
\[
r_{E,ij}=e_j|_{U_{ij}}-e_i|_{U_{ij}}.
\]
By the same torsor argument as in \S\ref{sec:42}, $\rE$ is a cocycle and
its class does not depend on the choice of atlas.

\begin{proposition}[equation semantic realization: construction of $\chi^E$]\label{prop:chiE}
To each supported semantic atom $\lambda\in S(V)$ on each cover
intersection $V$, assign a required equation index $i_\lambda$ and a
local architecture reading $A_\lambda$, compatibly with restriction
(selected; for a face restriction $V'\to V$ of the cover intersection
diagram, $i_{\lambda|_{V'}}=i_\lambda$ and
$A_{\lambda|_{V'}}=A_\lambda$). Then
\[
\chi^E_V(\lambda):=[\epsilon_{V,A_\lambda,i_\lambda,\pi_V(\lambda)}]\in \QE(V)
\]
is restriction-natural with respect to face restrictions, i.e.
\[
\chi^E_{V'}(\lambda|_{V'})=\chi^E_V(\lambda)|_{V'}.
\]
\end{proposition}

\begin{proof}
The indices $i_\lambda$ and readings $A_\lambda$ are chosen compatibly
with restriction, and $\pi$ preserves the underlying Atom
(\S\ref{sec:42}). Hence the residual at $\lambda|_{V'}$ is the
$\epsilon$ evaluated on $V'$ with the same reading, index, and support
Atom, and by the residual restriction naturality of
Theorem~\ref{thm:quotient} its class coincides with the restriction of
$\chi^E_V(\lambda)$.
\end{proof}

$\chi^E$ is the canonical instance, constructed from the equation system
and displayed readings, of the restriction-natural map that assumption~3
of Theorem~\ref{thm:central} requires; it shows that the starting point
of the comparison map is generated from nothing more than ``which
failure, of which reading of which equation''.

\section{The SAGA Comparison Theorem}\label{sec:5}

This section states and proves the central theorem of SAGA. Routine
verifications (such as componentwise checks of restriction naturality)
are compressed in the text, but every main argument that makes the
comparison theorem hold can be followed continuously within this
section.

\subsection{The central theorem}\label{sec:51}

Before the formal statement, we describe the shape of the claim in
plain terms. The stage is a selected monomorphic AAT cover, on which
stand the two complexes constructed independently in
Section~\ref{sec:4} --- the semantic side in the language of repair and
the equation side in the language of equations. The theorem asserts
that, under finitely many selection conditions matching the local data
of the two sides, (i) the two coefficient presheaves become naturally
isomorphic; (ii) this isomorphism identifies the two $H^1$ groups and
sends the obstruction class of the repair side to the obstruction class
of the equation side; and (iii) if moreover the family of repair states
satisfies the gluing condition (the true sheaf condition), then the
existence of a global repair is equivalent to the vanishing of the
obstruction class --- three stages in all. The inputs 1--8 below
(referred to interchangeably as assumptions 1--8) can be read in four
groups: inputs 1--2 are the two coefficients; inputs 3--4
are the correspondence connecting the coefficients and its
completeness; inputs 5--7 are the local state systems generating the
residuals and their correspondence; input 8 is the normalization of
empty overlaps.

\begin{theorem}[SAGA central theorem]\label{thm:central}
On a monomorphic AAT cover $\cU$, fix the following:
\begin{enumerate}
\item the abelian-group-valued presheaf $\Msem$ generated from
supported semantic atoms and a restriction-stable local repair relation
(\S\ref{sec:42});
\item the abelian-group-valued presheaf $\QE$ generated from an
architectural equation system $E$;
\item a restriction-natural map sending supported semantic atoms to
displayed Atom/equation residual classes;
\item completeness of the local repair relation with respect to that
map, and equation-generator completeness (checked on each cover
intersection);
\item an affine semantic repair system $\Psem$ deriving an additive
torsor structure from the action of repair words, together with its
selected local repair atlas;
\item an equation-lift system $\PE$ acted on by $\QE$, together with
its selected local lift atlas;
\item a restriction-natural, generator-equivariant map
$\beta:\Psem\to\PE$ sending semantic local states to equation local
lifts (checked on each cover intersection);
\item for each empty cover intersection $V$ excluded from the products,
$\Msem(V)=\QE(V)=0$ and the subsingleton property of $\Psem(V)$ and
$\PE(V)$ (empty-overlap normalization).
\end{enumerate}
Then the following three stages hold.

\textbf{(i) SAGA Presentation Theorem.} The free extension
$\widetilde\chi_V:\Fsem(V)\to \QE(V)$ of the map of assumption~3
(\S\ref{sec:53}) satisfies, on each nonempty cover intersection $V$,
\[
\ker \widetilde\chi_V=\Rrep(V),
\qquad
\operatorname{im}\widetilde\chi_V=\QE(V).
\]
Of the kernel equality, the inclusion
$\Rrep(V)\subseteq\ker\widetilde\chi_V$ (repair-relation soundness) is
not an assumption: it is derived from the local-state data
(assumptions~5--7) in \S\ref{sec:53}. The reverse inclusion
$\ker\widetilde\chi_V\subseteq\Rrep(V)$ is repair-relation
completeness, and the image equality is equation-generator
completeness; both are assumption~4. By this exactness, the map induces
a natural isomorphism on the cover intersection diagram
\[
\Phi:\Msem\xrightarrow{\ \sim\ }\QE.
\]

\textbf{(ii) \v{C}ech comparison.} $\Phi$ induces a degreewise cochain
isomorphism between the \v{C}ech complexes built separately from
$\Msem$ and $\QE$,
\[
\kappa^\bullet:
C^\bullet_{\mathrm{sem}}(\mathcal U)
\xrightarrow{\ \sim\ }
C^\bullet_E(\mathcal U),
\]
and $\kappa$ commutes with the differentials. Hence an isomorphism
\[
\kappa_*:
H^1_{\mathrm{sem}}(\mathcal U)
\xrightarrow{\ \sim\ }
\check H^1(\mathcal U,\QE)
\]
is induced. For the residual $\rsem$ generated from the semantic local
repair atlas and the residual $\rE$ generated from the equation
local lift atlas,
\[
\kappa_*([\rsem])=[\rE].
\]
Because the selected local atlases restrict to each nonempty
intersection, the local-state systems used by this construction are
nonempty on each intersection. The correspondence of residuals holds at the
cochain level, $\kappa^1(\rsem)=\rE$, when the two local atlases are
chosen aligned via $\beta$; when they are chosen independently, the
difference of the two cochains is an explicit $\delta^0$-image, and the
equality of classes holds.

\textbf{(iii) Grounded Global Gluing.} If moreover the $\Psem$ of
assumption~5 is a true semantic repair sheaf (defined in
\S\ref{sec:56}), then
\[
\mathrm{Nonempty}\,\Psem(W)
\iff
[\rsem]=0
\iff
[\rE]=0.
\]
\end{theorem}

\paragraph{The character of the three stages.}
In (ii), the cochain isomorphism and the induced $H^1$ isomorphism are
the general mechanism of \v{C}ech complexes that follows from a natural
isomorphism of coefficient presheaves. The latter half of (ii), the
residual correspondence, is part of the specific content of the theorem.
That specific content is concentrated in (i), (iii), and the residual
correspondence: the independent construction of the semantic repair
presentation and the equation-generated quotient presentation
(Section~\ref{sec:4}); the construction of the generator map $\chi$
(Proposition~\ref{prop:chiE}); the derivation of relation soundness
from the local-state interpretation $\beta$ (\S\ref{sec:53}); the
placement of completeness / generation as conditions relativized to
architecture data; and the semantic correspondence of residual classes
(\S\ref{sec:55}).

\subsection{The structure of the proof}\label{sec:52}

The proof is organized in four steps, and
\S\ref{sec:53}--\S\ref{sec:56} give the arguments in this order:
\begin{center}
\small
\begin{tabular}{ll}
generator map $\chi$ $\to$ coefficient isomorphism
$\Phi : \Msem \simeq \QE$ & (\S\ref{sec:53})\\
$\to$ cochain isomorphism $\kappa$, $\kappa\delta=\delta\kappa$ &
(\S\ref{sec:54})\\
$\to$ $H^1$ isomorphism $\kappa_*$ & (\S\ref{sec:54})\\
$\to$ residual correspondence $\kappa_*([\rsem])=[\rE]$ &
(\S\ref{sec:55})\\
$\to$ grounded global gluing & (\S\ref{sec:56})\\
\end{tabular}
\end{center}

That soundness is a consequence rather than an assumption is the
structural crux of the theorem; its derivation opens \S\ref{sec:53}.

The comparison core (\S\ref{sec:53}--\S\ref{sec:55}) uses only local
data on cover intersections; the global sheaf condition does not appear
until the gluing of \S\ref{sec:56}. Which proof consumes which
assumption of Theorem~\ref{thm:central} is fixed as a table in
Appendix~\ref{app:A}.

The overall structure of the proof is shown in
Figure~\ref{fig:comparison}.

\begin{figure}[htbp]
\centering
\includegraphics[width=0.92\linewidth]{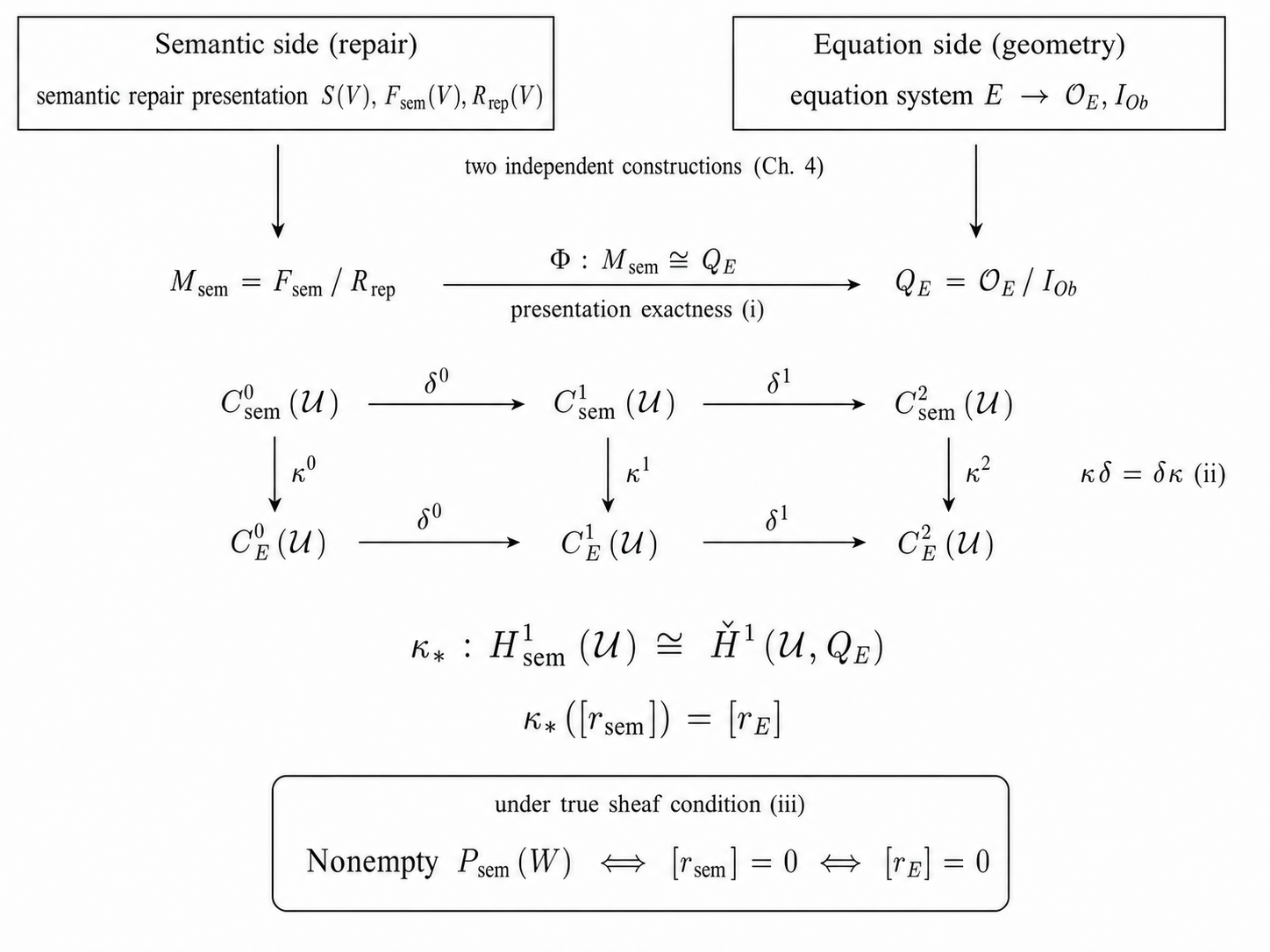}
\caption{The structure of the SAGA comparison theorem
(Theorem~\ref{thm:central}). The semantic side (upper left) and the
equation side (upper right), constructed independently in
Section~\ref{sec:4}, pass through the coefficient isomorphism $\Phi$
(i) and the cochain isomorphism $\kappa$ commuting with the
differentials (ii) to the $H^1$ isomorphism $\kappa_*$ and the residual
class correspondence $\kappa_*([\rsem])=[\rE]$; under the true sheaf
condition the three-way equivalence for global repair (iii) holds.}
\label{fig:comparison}
\end{figure}

\subsection{Coefficient isomorphism: generation of soundness and presentation comparison}\label{sec:53}

The map of assumption~3 is a restriction-natural correspondence
\[
\chi_V:S(V)\longrightarrow \QE(V)
\]
sending supported semantic atoms to equation coefficients on each
nonempty cover intersection $V$ (the $\chi^E$ of
Proposition~\ref{prop:chiE} is a constructed instance of this input).
Since $\Fsem(V)$ is the free abelian group on $S(V)$ (\S\ref{sec:42}),
by the universal property $\chi_V$ extends to a unique homomorphism
$\widetilde\chi_V:\Fsem(V)\to \QE(V)$. The $\Phi$ to be constructed is
the map induced by $\widetilde\chi$ on the quotient by the repair
relation, and its well-definedness, injectivity, and surjectivity are
carried by three distinct conditions: soundness, completeness, and
generation. The joint validity of these three conditions is called
\textbf{SAGA presentation exactness}.

\paragraph{Generation of soundness.}
Take $x\in\Rrep(V)$ and $p\in\Psem(V)$ (under
Theorem~\ref{thm:central}, $\Psem(V)$ is nonempty by restriction of the
selected local atlas). By the action soundness of \S\ref{sec:42}, the
semantic action gives $x+p=p$. The generator-level equivariance of
$\beta$ extends to finite sums and additive inverses in the free group,
so
\[
\beta(p)=\beta(x+p)=\widetilde\chi(x)+\beta(p).
\]
Since the action of $\QE(V)$ on $\PE(V)$ is free, $\widetilde\chi(x)=0$.
Therefore
\[
\Rrep(V)\subset\ker\widetilde\chi_V.
\]
Deriving soundness from the local-state data instead of assuming it is
exactly this three-line argument.

The definition of $\beta$ in assumption~7 does not mention $\Rrep$;
all that is required is restriction naturality and generator-level
equivariance, so the derivation is not circular. That $\beta$ can be
constructed independently of soundness is shown by the concrete map in
the finite example of \S\ref{sec:58}.

\paragraph{Descent to the quotient.}
Since $\Rrep(V)\subset\ker\widetilde\chi_V$, the map
$\widetilde\chi_V$ factors through the quotient, inducing a unique
homomorphism
\[
\Phi_V:
\Msem(V)=\Fsem(V)/\Rrep(V)
\longrightarrow \QE(V).
\]

\paragraph{Injectivity.}
If $\Phi_V([x])=0$ then $x\in\ker\widetilde\chi_V$, and by
repair-relation completeness (assumption~4) $x\in\Rrep(V)$, hence
$[x]=0$.

\paragraph{Surjectivity.}
For any $q\in \QE(V)$, equation-generator completeness (assumption~4)
provides $x\in\Fsem(V)$ with $\widetilde\chi_V(x)=q$, and
$\Phi_V([x])=q$.

\paragraph{Naturality.}
On generators,
\[
\Phi_{V'}([\lambda|_{V'}])
  =\chi_{V'}(\lambda|_{V'})
  =\chi_V(\lambda)|_{V'}
  =\Phi_V([\lambda])|_{V'},
\]
where the second equality is the restriction naturality of $\chi$
(assumption~3). Since the generators generate $\Msem(V)$, $\Phi$
commutes with restriction on all elements.

Thus $\Phi:\Msem\simeq \QE$ is a natural isomorphism of presheaves on
the cover intersection diagram. This isomorphism is not a transport of
one complex from the other but a presentation comparison between the
semantic presentation and the equation quotient. The three conditions
carry three different mathematical jobs --- existence, injectivity, and
surjectivity of the map --- and removing any one of them destroys the
isomorphism, as counterexamples over the free group on one generator
$F=\Zint$ confirm:
\begin{itemize}
\item \emph{dropping soundness}: $R=2\Zint$, $Q=\Zint$, $\chi(n)=n$.
The relation $2=0$ does not become zero in $Q$, and the map does not
descend to the quotient.
\item \emph{dropping completeness}: $R=0$, $Q=\Zint/(2)$,
$\chi(n)=[n]$. The map descends and is surjective, but $2$ remains in
the kernel and the map is not injective.
\item \emph{dropping generation}: $R=2\Zint$, $Q=\Zint/(4)$,
$\chi(n)=[2n]$. The kernel coincides with $R$, but the image is
$\{0,[2]\}$ and the map is not surjective.
\end{itemize}

\subsection{Cochain commutation and the \texorpdfstring{$H^1$}{H1} isomorphism}\label{sec:54}

\paragraph{Construction of the degreewise map.}
For each degree $n=0,1,2$, define
\[
\kappa^n:C^n_{\mathrm{sem}}(\mathcal U)\longrightarrow C^n_E(\mathcal U)
\]
by applying $\Phi$ to each intersection component:
\[
(\kappa^0a)_i=\Phi_{U_i}(a_i),\qquad
(\kappa^1c)_{ij}=\Phi_{U_{ij}}(c_{ij}),\qquad
(\kappa^2z)_{ijk}=\Phi_{U_{ijk}}(z_{ijk}).
\]
Since each $\Phi_V$ is an isomorphism, $\kappa^n$ is an isomorphism on
the product.

\paragraph{Commutation with the differentials.}
What must be shown is the commutativity of the diagram
\[
\begin{array}{ccccc}
C^0_{\mathrm{sem}}(\mathcal U)
&\xrightarrow{\ \delta^0_{\mathrm{sem}}\ }&
C^1_{\mathrm{sem}}(\mathcal U)
&\xrightarrow{\ \delta^1_{\mathrm{sem}}\ }&
C^2_{\mathrm{sem}}(\mathcal U)\\[2pt]
\big\downarrow{\scriptstyle\kappa^0}&&\big\downarrow{\scriptstyle\kappa^1}&&\big\downarrow{\scriptstyle\kappa^2}\\[2pt]
C^0_E(\mathcal U)
&\xrightarrow{\ \delta^0_E\ }&
C^1_E(\mathcal U)
&\xrightarrow{\ \delta^1_E\ }&
C^2_E(\mathcal U)
\end{array}
\]
For $a\in C^0_{\mathrm{sem}}(\cU)$ and $i<j$,
\[
\begin{aligned}
(\kappa^1\delta^0_{\mathrm{sem}}a)_{ij}
&=\Phi_{U_{ij}}(a_j|_{U_{ij}}-a_i|_{U_{ij}})\\
&=\Phi_{U_j}(a_j)|_{U_{ij}}-\Phi_{U_i}(a_i)|_{U_{ij}}\\
&=(\delta_E^0\kappa^0a)_{ij},
\end{aligned}
\]
where the second line is the additivity and restriction naturality of
$\Phi$. Applying the same computation to the $(i,j,k)$ component of
$c\in C^1_{\mathrm{sem}}(\cU)$ (the alternating three-term sum) yields
$\kappa^2\delta^1_{\mathrm{sem}}=\delta^1_E\kappa^1$. Only the
additivity and naturality of $\Phi$ proved in \S\ref{sec:53} are used;
no new assumption enters here.

\paragraph{$H^1$ isomorphism.}
Since $\kappa$ is a degreewise isomorphism commuting with the
differentials, it preserves and reflects cocycles and coboundaries.
Hence
\[
\kappa_*:
H^1_{\mathrm{sem}}(\mathcal U)
\longrightarrow
\check H^1(\mathcal U,\QE),
\qquad
[c]\longmapsto[\kappa^1c]
\]
is well defined: replacing a representative by
$c+\delta^0_{\mathrm{sem}}a$ gives
$\kappa^1(c+\delta^0_{\mathrm{sem}}a)=\kappa^1c+\delta^0_E(\kappa^0a)$, and the class
does not change. The inverse map is induced in the same form by the
degreewise inverse $(\kappa^1)^{-1}$, and $\kappa_*$ is an isomorphism
of abelian groups.

\subsection{Residual correspondence}\label{sec:55}

Choose the semantic atlas $\{p_i\}$ and the equation atlas $\{e_i\}$
independently. On each chart, $\beta(p_i)$ and $e_i$ are elements of
the same $\QE(U_i)$-torsor $\PE(U_i)$, so there is a unique
$h_i\in \QE(U_i)$ with $e_i=h_i+\beta(p_i)$. On $U_{ij}$, using the
naturality and equivariance of $\beta$ (assumption~7),
\[
\begin{aligned}
e_j
&=h_j+\beta(p_j)\\
&=h_j+\beta(r_{\mathrm{sem},ij}+p_i)\\
&=h_j+\Phi(r_{\mathrm{sem},ij})+\beta(p_i)\\
&=(\Phi(r_{\mathrm{sem},ij})+h_j-h_i)+e_i.
\end{aligned}
\]
By uniqueness of differences in the equation torsor,
\[
r_{E,ij}=\Phi(r_{\mathrm{sem},ij})+h_j-h_i,
\qquad\text{that is,}\qquad
\rE=\kappa^1(\rsem)+\delta_E^0h.
\]
When the two atlases are chosen aligned via $\beta$, i.e.\ when
$e_i=\beta(p_i)$, we have $h=0$ and the cochain-level equality
$\kappa^1(\rsem)=\rE$ holds. When they are chosen independently, the
difference of the two cochains is the explicit coboundary
$\delta^0_Eh$, so the equality of classes
\[
\kappa_*([\rsem])=[\rE]
\]
holds. Since $\kappa_*$ is an isomorphism and preserves and reflects
zero, the equivalences $[\rsem]=0\iff[\rE]=0$ and
$[\rsem]\neq0\iff[\rE]\neq0$ follow as well.

This proves parts (i) and (ii) of Theorem~\ref{thm:central}
(\S\ref{sec:53}--\S\ref{sec:55}). The remaining three-way equivalence
(iii) is proved in the next subsection.

\subsection{Global repair: Grounded Global Gluing}\label{sec:56}

\begin{definitionstar}[true semantic repair sheaf]
A presheaf $\Psem$ on $S_X$ is called a \textbf{true semantic repair
sheaf} when: (1) $\Psem$ satisfies the sheaf condition for all covers
of the selected AAT topology; (2) the action of $\Msem$ commutes with
restriction and is locally free and transitive; (3) $\cU$ is a
monomorphic AAT cover belonging to that topology; (4) on the empty
intersections excluded from the products, $\Psem$ is subsingleton.
From these four conditions the sheaf condition for $\cU$ is derived,
and no per-cover amalgamation map is posited as separate data.
\end{definitionstar}

\begin{lemmafiveA}[ordered matching completion]
Suppose each morphism of the cover is a monomorphism and $\Psem$ is
subsingleton on the empty overlaps excluded from the products
(empty-overlap normalization). If a family $(q_i)$,
$q_i\in\Psem(U_i)$, agrees on all nonempty $i<j$ overlaps, then
$(q_i)$ is a matching family on all ordered overlaps.
\end{lemmafiveA}

\begin{proof}
Three kinds of overlap remain. On a self-overlap $U_{ii}$: since $u_i$
is a monomorphism, the diagonal $U_i\to U_i\times_W U_i$ is an
isomorphism, and under this canonical isomorphism the two projections
become the same morphism; hence the two restrictions of $q_i$ agree.
On a reversed overlap $U_{ji}$ ($i<j$): the pullback symmetry
$U_j\times_W U_i\cong U_i\times_W U_j$ is a canonical isomorphism
interchanging the two projections; restrictions are compatible along
this canonical isomorphism, so agreement on $U_{ji}$ is equivalent to
agreement on $U_{ij}$. On the empty overlaps excluded from the
products, $\Psem$ is subsingleton, so agreement is automatic.
\end{proof}

\begin{theorem}[Grounded Global Gluing]\label{thm:gluing}
In the setting of Theorem~\ref{thm:central}, if in addition $\Psem$ is
a true semantic repair sheaf, then the three-way equivalence
\[
\mathrm{Nonempty}\,\Psem(W)
\iff
[\rsem]=0
\iff
[\rE]=0
\]
holds. Moreover, for the correction $a$ obtained from a proof of
$[\rsem]=0$, there exists a unique global section whose restrictions
are the corrected family
\[
p_i^{\mathrm{corr}}=(-a_i)+p_i.
\]
\end{theorem}

\begin{proof}[Proof (forward)]
Suppose $[\rsem]=0$. Since the zero class is equivalent to the
existence of a matching correction (\S\ref{sec:42}), there is
$a\in C^0_{\mathrm{sem}}(\cU)$ with $\rsem=\delta^0_{\mathrm{sem}}a$.
The residual of the corrected family $p_i^{\mathrm{corr}}=(-a_i)+p_i$
is $\rsem-\delta^0_{\mathrm{sem}}a=0$ by choice independence (\S\ref{sec:42}), so the
corrected family agrees on all nonempty $i<j$ overlaps. By
Lemma~5.2A, the corrected family is a matching family on all ordered
overlaps. Since $\cU$ is a cover of the topology and $\Psem$ is a sheaf
for that topology, sheaf amalgamation provides a unique
$p\in\Psem(W)$ with $p|_{U_i}=p_i^{\mathrm{corr}}$.
\end{proof}

\begin{proof}[Proof (reverse)]
Conversely, suppose some $p\in\Psem(W)$ exists. On each chart, by
transitivity and freeness of the torsor there is a unique
$a_i\in\Msem(U_i)$ with $p_i=a_i+p|_{U_i}$. Taking differences on
overlaps,
\[
r_{\mathrm{sem},ij}=a_j-a_i,
\qquad\text{that is,}\qquad
\rsem=\delta^0_{\mathrm{sem}}a,
\]
so $[\rsem]=0$. The final equivalence of the three,
$[\rsem]=0\iff[\rE]=0$, is the residual class correspondence of
\S\ref{sec:55}.
\end{proof}

What is unique here is the amalgamation of the fixed corrected
matching family, not uniqueness of $\Psem(W)$ as a whole; the
difference between distinct global repairs is carried by degree-zero
data.

\paragraph{The equation side.}
The same conclusion descends to the equation side. Three additional
assumptions are made: SAGA presentation exactness holds on all
contexts of $S_X$; the $\beta$ of assumption~7 is given as a natural
transformation on the whole site (the comparison core uses only its
components on the intersection diagram, so this globalization is an
additional assumption); and $\PE$ is also a sheaf for the selected
topology. On each chart and nonempty overlap, the
$\Phi_V$-equivariant $\beta_V:\Psem(V)\to\PE(V)$ is a bijection:
injectivity follows from the transitivity of the semantic torsor, the
freeness of the target torsor, and the injectivity of $\Phi_V$;
surjectivity from the transitivity of the target torsor and the
surjectivity of $\Phi_V$. This local bijectivity lifts to the global
level in three steps. Injectivity: if $\beta_W(p)=\beta_W(p')$, then
on each chart the injectivity of $\beta_{U_i}$ gives
$p|_{U_i}=p'|_{U_i}$, and the separatedness of $\Psem$ gives $p=p'$.
Surjectivity: for $q\in\PE(W)$, set
$p_i:=\beta_{U_i}^{-1}(q|_{U_i})$ on each chart; by the naturality of
$\beta$ and injectivity on overlaps, the family $(p_i)$ agrees on all
nonempty $i<j$ overlaps, and by Lemma~5.2A and the amalgamation of
$\Psem$ it glues to a unique $p\in\Psem(W)$. Finally, $\beta_W(p)$
and $q$ have equal restrictions to each chart, so the separatedness
of $\PE$ gives $\beta_W(p)=q$. Thus
\[
\mathrm{Nonempty}\,\PE(W)\iff[\rE]=0
\]
follows.

\subsection{Scope of the theorem}\label{sec:57}

This comparison theorem is an additive theorem for
abelian-group-valued presheaves. The pointed-set-valued $H^1$ of
nonabelian torsors, higher coherence via 2-cocycles and gerbes, and
stack descent each require independent statements. None of these
conclusions is derived from the additive $H^1$ comparison.

Moreover, the $H^1$ treated by this theorem is the \v{C}ech
$\check H^1(\cU,-)$ relative to the selected monomorphic AAT cover $\cU$.
Identification with cover-independent sheaf cohomology requires the
additional conditions of refinement invariance or Leray-type
acyclicity, and this paper does not claim it.

Furthermore, the Lean formalization of this theorem is carried out
over a site whose context category is thin (the finite-meet poset
model of \S\ref{sec:33} is the representative example). Over general
$\ArchCtx(X)$, the gluing argument of \S\ref{sec:56}, including
Lemma~5.2A, uses monomorphisms and the universal property of
pullbacks, but that version has not been formalized (the details of
the hypotheses that thinness renders unnecessary are in
Appendix~\ref{app:A4}).

\subsection{Finite witness: independently generated circle comparison}\label{sec:58}

\begin{example}\label{ex:circle}
Take a monomorphic 4-cycle cover with four charts and four nonempty
overlaps $U_{01}, U_{12}, U_{23}, U_{03}$, and no nondegenerate triple
overlap.

\paragraph{Comparison of coefficients.}
On each nonempty intersection $V$, take the semantic support to be a
single generator $\sigma_V$ and the local repair relation to be
$2\sigma_V=0$ (writing $\Zint[\sigma_V]$ for the free abelian group on
the generator $\sigma_V$); the semantic side independently generates
\[
\Msem(V)=\Zint[\sigma_V]/(2\sigma_V)\cong\Ftwo.
\]
On the equation side, $\ObsE(V)=\Zint$ and $\IOb(V)=(2)$ give
$\QE(V)=\Zint/(2)$. For the interpretation $\chi_V(\sigma_V)=[1]$,
the three conditions of \S\ref{sec:53} can be checked directly:
\begin{center}
\small
\begin{tabular}{ll}
soundness: & $2\sigma_V \mapsto [2] = 0$\\
completeness: & $\ker(\Zint \to \Zint/(2)) = 2\Zint$\\
generation: & $[1]$ generates $\Zint/(2)$\\
\end{tabular}
\end{center}
Hence the construction of \S\ref{sec:53} gives $\Msem\simeq \QE$.
This is not a transport of one complex but a presentation comparison
between the semantic presentation $\Zint[\sigma_V]/(2\sigma_V)$ and the
equation quotient $\Zint/(2)$.

\paragraph{Local state system and affine transitions.}
On each chart and each nonempty overlap put $\Psem=\Ftwo$, with
$\Msem(V)\cong\Ftwo$ acting by addition. For an oriented edge $e$
(source chart $\to$ target chart), take the two restrictions to the
overlap to be
\[
\rho^{\mathrm{src}}_e(x)=x,
\qquad
\rho^{\mathrm{tgt}}_e(x)=x+t_e,
\]
where the transitions are
\[
(t_{01},t_{12},t_{23},t_{30})=(1,0,0,0).
\]
The last edge is $U_{03}$ oriented $3\to0$; in $\Ftwo$ the sign under
orientation reversal does not change anything. Since there is no
triple overlap, no additional functoriality condition constrains
these restrictions.

\paragraph{Computation of the residual.}
On each chart choose the local state $p_i=0$. On the overlap of an
edge $e$, the same $0$ is restricted from the source chart and from
the target chart, but the two restrictions are different maps. By the
definition of the residual in \S\ref{sec:42} (the difference of the
two restrictions on the overlap),
\[
r_{\mathrm{sem},e}
=\rho^{\mathrm{tgt}}_e(0)-\rho^{\mathrm{src}}_e(0)
=t_e.
\]
Even with all local states equal to $0$, the residual does not
vanish: what carries the mismatch is not the values of the states but
the difference of the two restrictions carrying the same $0$ to the
overlap --- the affine transition. Therefore
\[
\rsem=(1,0,0,0)\in\Ftwo^4.
\]

\paragraph{The equation side and the independent construction of $\beta$.}
Put $\PE=\Ftwo$ and give the affine transitions on oriented edges
independently by the same numbers $(1,0,0,0)$. The local-state
interpretation $\beta$ is the identity map on the $\Ftwo$ coordinates
on each chart and each nonempty overlap. This definition mentions
neither $\Phi$ nor the repair relation. Equivariance on generators
is, under $\chi(\sigma)=[1]$, the direct $\Ftwo$ computation
$\beta(\sigma+x)=x+1=\chi(\sigma)+\beta(x)$, and commutation with
restriction is a componentwise computation from the fact that the
transitions on both sides are given by the same numbers. Thus the
data of assumptions~5--7 are assembled independently of $\Rrep$, and
the derivation of \S\ref{sec:53} yields soundness
$\widetilde\chi(2\sigma)=[2]=0$ --- agreeing with the soundness row
checked directly in the coefficient comparison above. Choosing
$e_i=0$ on each chart --- which is the aligned choice
$e_i=\beta(p_i)$ of \S\ref{sec:55} --- gives, at the cochain level,
\[
\rE=\kappa^1(\rsem)=(1,0,0,0).
\]

\paragraph{Cocycle property and non-coboundary property.}
Since there is no nondegenerate triple overlap, $C^2(\cU)=0$ and
$\delta^1\rsem=0$ hold automatically. On the other hand, for any
$a\in C^0(\cU)$ the components of $\delta^0a$ are the
$a_j|_{U_{ij}}-a_i|_{U_{ij}}$ on the overlaps, and taking the edge sum around the
4-cycle, each chart's component appears exactly twice, so in $\Ftwo$
the sum is zero. The edge sum of $\rsem$ is $1+0+0+0=1$, so $\rsem$
is not a $\delta^0$-image. Therefore
\[
[\rsem]\ne0,\qquad
[\rE]\ne0,\qquad
\kappa_*([\rsem])=[\rE].
\]

\paragraph{Relation to real code.}
This example is a finite witness that the SAGA comparison preserves
not only zero classes but also nonzero classes. What erects the
nonzero class is not the number of charts but the
\textbf{cycle-without-a-face mechanism}: odd parity on a closed loop
plus the absence of a face (triple overlap) filling that loop. The
one-cent obstruction of Section~\ref{sec:7} is an instance of this
mechanism appearing in a real architecture (\S\ref{sec:71}).
\end{example}

\section{Lean Formalization Status}\label{sec:6}

The Lean formalization records the machine-checked definitions,
theorems, witnesses, and proof chain of the SAGA mathematics. This
section reports the formalization status at release time at declaration
granularity. The primary evidence for the status consists of the Lean
sources fixed by the release tag; the kernel axiom audit, which checks
every registered declaration against the allowlist of standard Mathlib
axioms --- \code{propext}, \code{Classical.choice}, \code{Quot.sound};
and the focused check of the release CI (the organization of the
sources and the audit mechanism are in Appendix~\ref{app:B}). Every
declaration listed in the table of this section is registered in this
audit.

\subsection{Scope of the formalization}\label{sec:61}

The SAGA theorem chain is formalized by fixing the input structure of
the theorem of Section~\ref{sec:5} (inputs 1--8 of
Theorem~\ref{thm:central}) directly as Lean structures. Accordingly,
what each declaration proves is not the theorem of this paper as such,
but the conclusion over the selected inputs fixed as Lean structures
(the correspondence with the paper claims and the assumptions of each
row are fixed by the table of \S\ref{sec:62}). The building blocks are
as follows (the correspondence with source files is in
Appendix~\ref{app:B}):
\begin{itemize}
\item the monomorphic ordered cover, the intersection diagram, the
three-term \v{C}ech complex, and the cover-relative $H^1$;
\item the semantic repair presentation and the presheaf construction
of $\Msem$;
\item the affine semantic repair system, the selected atlases on both
sides, and the residuals;
\item the derivation of soundness, presentation exactness, the
coefficient isomorphism $\Phi$, and the finite counterexamples exhibiting
the failure of each of the three conditions;
\item the equation-side realization / production (the connection to
the $\QE$ generated by the equation system);
\item the cochain comparison $\kappa$, the $H^1$ isomorphism, and the
residual correspondence;
\item true sheaf descent, grounded global gluing, and the final
bundling of the central theorem;
\item the finite witnesses (the nonzero 4-cycle circle witness and the
zero-class descent witness);
\item the site realization over the finite-meet poset model of
\S\ref{sec:33} and the comparison bridge to the ordered \v{C}ech
complex of \S\ref{sec:41} (including the formalized counterpart of
Lemma~5.2A).
\end{itemize}

The docstrings of the Lean sources use labels from a numbering series
distinct from the theorem numbers of this paper (normalized here as
\code{X.Theorem 1.1} and so on). The correspondence between declarations and labels is also
fixed in Appendix~\ref{app:B}.

\subsection{Status table}\label{sec:62}

Every status in Table~\ref{tab:status} is \code{proved} (proved in Lean).

\begin{table}[htbp]
\centering
\scriptsize
\begin{tabular}{@{}p{0.27\textwidth}p{0.30\textwidth}p{0.06\textwidth}p{0.28\textwidth}@{}}
\toprule
Paper claim & Lean declaration & Status & Assumptions \\
\midrule
Conclusion bundle of Theorem~\ref{thm:central} (residual correspondence, zero/nonzero
equivalence, grounded gluing) &
\code{SagaEquationPacket.}\allowbreak\code{sagaCentralTheorem} & proved &
selected packet (\S\ref{sec:61}), the two completeness conditions
(input~4), \code{Fintype} of the cover index set. The gluing clause is
additionally conditional on the cover belonging to the topology and the
true sheaf condition \\
\addlinespace
Theorem~\ref{thm:central}(i): derivation of repair-relation soundness &
\code{PrimaryStateCorrespondence.}\allowbreak\code{relationSound\_}\allowbreak\code{of\_stateCorrespondence} &
proved & local-state correspondence (inputs 5--7) \\
\addlinespace
Theorem~\ref{thm:central}(i): coefficient isomorphism $\Phi$ &
\code{SagaEquationPacket.phiEquiv} (generic version:
\code{PrimaryCoefficientCorrespondence.}\allowbreak\code{phiEquiv}) & proved &
soundness and the two completeness conditions \\
\addlinespace
Theorem~\ref{thm:central}(i): counterexamples removing each of the three conditions &
\code{ExactnessFixtures.}\allowbreak\code{soundness\_failure} /
\code{completeness\_failure} / \code{generation\_failure} & proved &
none (each fixture is a finite counterexample stating the failure of
the targeted condition; simultaneous validity of the remaining two
conditions is not part of the statement) \\
\addlinespace
Theorem~\ref{thm:central}(ii): cochain commutation $\kappa\delta=\delta\kappa$ &
\code{kappa1\_delta0}, \code{kappa2\_delta1} & proved &
restriction-natural isomorphism of the coefficient families \\
\addlinespace
Theorem~\ref{thm:central}(ii): $H^1$ isomorphism $\kappa_*$ &
\code{SagaEquationPacket.}\allowbreak\code{kappaStarAddEquiv} ($\simeq_+$,
additive equivalence, on the packet; generic version:
\code{kappaH1AddEquiv}) & proved & the two
completeness conditions \\
\addlinespace
Theorem~\ref{thm:central}(ii): residual correspondence $\kappa_*([\rsem])=[\rE]$ &
\code{SagaEquationPacket.}\allowbreak\code{residual\_correspondence\_class};
aligned-atlas version \code{betaAligned\_residual} & proved & the two
completeness conditions \\
\addlinespace
Theorem~\ref{thm:gluing} = Theorem~\ref{thm:central}(iii): grounded global gluing &
\code{SagaEquationPacket.}\allowbreak\code{globalRepair\_nonempty\_iff}
($\Psem$-side equivalence); \code{sagaGroundedGluing} (integrated up to
the equation-side equivalence) & proved & true sheaf condition, cover
belonging to the topology; the equation-side equivalence additionally
requires the two completeness conditions \\
\addlinespace
Lemma~5.2A (ordered matching completion) &
\code{SiteStateData.matchingFamily\_iff} & proved & thin context
category (\S\ref{sec:57}, Appendix~\ref{app:A4}); subsingleton
assumption for states on the omitted pairs (the formalized counterpart
of empty-overlap normalization) \\
\addlinespace
Example~\ref{ex:circle} (4-cycle circle witness, transfer of the nonzero class) &
\code{CircleWitness.}\allowbreak\code{semanticResidualClass\_ne\_zero},
\code{circle\_nonzero\_class\_transfer} & proved & none (closed
verification on the concrete 4-cycle model) \\
\addlinespace
Zero-class witness (nonempty firing of Theorem~\ref{thm:gluing}) &
\code{DescentWitness.descentTrueSheaf},
\code{descent\_sagaGroundedGluing} & proved & none (closed verification
on a concrete model) \\
\bottomrule
\end{tabular}
\caption{Formalization status at declaration granularity.}
\label{tab:status}
\end{table}

All rows of the table are registered in the kernel axiom audit, and the
kernel axioms are limited to the standard Mathlib axioms. No row
contains \code{sorry} or additional axioms.

What remains with respect to the mathematics of this paper
(Sections~\ref{sec:3}--\ref{sec:5}) --- to be distinguished from the
research outlook of Section~\ref{sec:9} --- is the following.
\begin{itemize}
\item The gluing argument of \S\ref{sec:56} over a general (non-thin)
context category, in the form that explicitly consumes monomorphisms
and the universal property of pullbacks: \code{unported} (the paper
proof exists but has not been ported to Lean; \S\ref{sec:57}).
\item Identification with sheaf cohomology via refinement invariance /
Leray-type acyclicity, nonabelian $H^1$, gerbes, stack descent: not
claimed by this paper, hence not among the proof obligations
(\S\ref{sec:57}).
\item Transporting the measurement runs of Section~\ref{sec:7} into
Lean (generating instantiations of Theorem~\ref{thm:central} from
measured packets): not done. The finite instantiation of
Theorem~\ref{thm:central} is carried by the witness rows of this table
(Example~\ref{ex:circle}), and the packets of Section~\ref{sec:7} carry
the finite checks shown by the condition matrix of Appendix~\ref{app:C3}. The
division of labor between the two is stated in \S\ref{sec:75}.
\end{itemize}

The release snapshot of this section is the repository state at the
release tag \code{saga-paper-}\allowbreak\code{v1.0.0}. The full Lean build, including
the axiom audit, runs in continuous integration on the tagged commit;
the commit hash and the reference to that CI run are recorded in
\code{MANIFEST.json} of the deposit bundle, which shares the release
identity of this paper (Appendix~\ref{app:C4}).

\section{ArchSig: Executable SAGA Diagnosis}\label{sec:7}

ArchSig is a measurement system (a Rust command-line tool) that
computes grounding, derived residuals, boundary membership, comparison
of run pairs, and gate verdicts from two families of inputs:
observation (ArchMap) and law / equations (LawPolicy, law surfaces,
MeasurementProfile), together with a repair plan that declares only
the selected complex (\S\ref{sec:74}). This section first presents a full SAGA diagnosis
of a real microservice architecture as a single computation
(\S\ref{sec:71}--\S\ref{sec:73}), then fixes the essentials of the
input contract on which the diagnosis stands (\S\ref{sec:74}) and the
boundary of the claims (\S\ref{sec:75}). The audit surface --- the
per-artifact input contract and the definition of the computation, the
supply process for observation, the kinds of conditions behind the
conclusions, and the reproduction procedure --- is fixed in
Appendix~\ref{app:C}.

\subsection{A real-code case: the one-cent obstruction}\label{sec:71}

The subject of the case study is the open-source train reservation
system
\href{https://github.com/FudanSELab/train-ticket}{train-ticket}
(commit \code{313886e99bef}, 42 services), widely used as a
microservice benchmark.

On the real call triangle cancel--inside-payment--order of this
system, the three services follow three different conventions for the
refund amount. All three were confirmed in the actual sources and observed as
section values of the ArchMap:
{\sloppy
\begin{itemize}
\item \textbf{cancel}: \code{Double.parseDouble(order.getPrice())}\allowbreak\code{ * 0.8}
rounded by \code{DecimalFormat("0.00")} and rendered as a string
(floating point + rounding);
\item \textbf{inside-payment}: exact arithmetic on
\code{new BigDecimal(order.getPrice())};
\item \textbf{order}: pass-through storage in
\code{private String price}.
\end{itemize}
}

The monetary representations of the three services are all passed as
\code{String} at the implementation level: the types agree. What
differs is the \emph{semantic convention} --- the rounding, scale,
computation, and storage of the amount that the \code{String}
represents. Furthermore, in the source investigation by the operator
(the person executing the measurement process; in this paper, the
author), no site was found that simultaneously reconciles the amounts
of the three services. This investigative finding is reflected in the
diagnosis in the form of declaring no triple overlap in the selected
complex. The absence of a simultaneous reconciliation site is itself
treated as an assertion of the operator, not as a fact exhibited by the
observation artifacts (\S\ref{sec:72}, Appendix~\ref{app:C3}).

This configuration can be read as a 3-cycle instance of the same
\textbf{cycle-without-a-face mechanism} as Example~\ref{ex:circle}.
The complexes are not identical --- Example~\ref{ex:circle} is a
4-cycle, this case a 3-cycle --- but the mechanism erecting the
nonzero class is the same: odd parity on a closed loop (the collision
of three conventions) and the absence of a face (triple overlap) filling
the loop. The rounding remainder of the refund computation
$0.8\times\text{price}$ --- a \textbf{drift below one cent} --- is
booked in no chart. We call this case study the \textbf{one-cent
obstruction}.

\paragraph{Semantic trace.}
Table~\ref{tab:semantictrace} fixes the quantities compared by the
three edges and the implementation conventions.

\begin{sidewaystable}[htbp]
\centering
\scriptsize
\begin{tabular}{@{}p{0.10\linewidth}p{0.09\linewidth}p{0.17\linewidth}p{0.13\linewidth}p{0.08\linewidth}p{0.15\linewidth}p{0.16\linewidth}@{}}
\toprule
Edge & Compared quantity & Convention L & Convention R & Normalization
& Expected equation & Witness variable \\
\midrule
cancel--inside-payment & refund amount & the string obtained by
rounding $0.8\times\text{price}$ with \code{DecimalFormat(}\allowbreak\code{"0.00")} &
the exact arithmetic value via \code{BigDecimal} & currency value
(scale-2) & the refund amounts of both charts are determined as the
same currency value & \code{e\_cancel\_insidepay} \\
\addlinespace
inside-payment--order & reference amount of the refund & the exact
value of \code{BigDecimal(}\allowbreak\code{order.}\allowbreak\code{getPrice())} & the pass-through
recorded amount of \code{String price} & currency value (scale-2) &
the reference amount used by payment is determined under the same
convention as the recorded amount & \code{e\_insidepay\_order} \\
\addlinespace
cancel--order & application of the refund ratio & the rounded
$0.8\times\text{price}$ & the recorded amount \code{price} & currency
value (scale-2) & the refund amount is determined uniquely as 0.8
times the recorded amount & \code{e\_cancel\_order} \\
\bottomrule
\end{tabular}
\caption{Semantic trace of the one-cent obstruction: the quantities
compared by the three edges and the implementation conventions.}
\label{tab:semantictrace}
\end{sidewaystable}

What this table fixes is the semantic-grounds side; the residual
values themselves are derived by \code{analyze} from the comparison of
the observed section values of each edge (Appendix~\ref{app:C1}). All three
edges compare restrictions of one and the same semantic quantity ---
the refund amount of this order --- and the degree of freedom of the
mismatch lies along a single direction: under which rounding / scale
convention that quantity is determined as a currency value. The
witness bindings (Appendix~\ref{app:C1}) are the law-side declarations
choosing that direction edge by edge.

\paragraph{Finite witness.}
Separated from any estimate of frequency or total loss (which requires
runtime measurement; \S\ref{sec:75}), we exhibit a witness with one
fixed input price, as a normalization computation applying each
source-confirmed convention to the same refund quantity:
\begin{center}
\footnotesize
\begin{tabular}{@{}ll@{}}
original price: & 12.33\\
cancel convention: & $0.8 \times 12.33 = 9.864$\\
 & $\to$ \code{DecimalFormat("0.00")} $\to$ ``9.86''\\
inside-payment convention: & the same refund quantity by exact
arithmetic\\
 & $\to$ 9.864\\
normalized as currency values: & 9.86 (scale-2) vs.\ 9.864\\
nonzero remainder: & 0.004 (below one cent)\\
\end{tabular}
\end{center}
The multiplication on the cancel side is performed in binary floating
point, but for this input that does not affect the result of the
scale-2 rounding. The remainder violates no chart's local equation ---
cancel is faithful to cancel's rounding convention, inside-payment to
exact arithmetic --- yet it remains as the difference of the values of
the two charts.

This numerical witness is a source-level recomputation by the
operator; none of the price, the $0.8$, the rounding, or the $0.004$
appears in ArchSig's computation.

\subsection{The diagnosis staircase}\label{sec:72}

The measurement inputs are composed as follows (the per-artifact input contract and
the definition of the computation are in Appendix~\ref{app:C1}).
\begin{itemize}
\item \textbf{cover}: 6 charts --- the diagnosis triangle
\{cancel, inside-payment, order\} and the consignment-fee region
\{preserve, consign, consign-price\}. The latter is an observation
region outside the triangle; it provides, within the same packet, a
contrast containing an agreeing edge and a mismatch that after repair
falls inside the boundary ($B^1$).
\item \textbf{law surfaces}: three --- closed-equational,
SAGA-grounded, and descent. The descent surface binds witness
variables on all six observed edges.
\item \textbf{repair plan}: declares only the selected complex. The
charts are exactly the 6 charts of the observed cover; the overlaps
are the 6 observed restriction edges (3 of the triangle +
consign--consign-price + preserve--consign + preserve--order); no
triple overlap is declared (reflecting the investigative finding of
\S\ref{sec:71}; its status as an assertion is discussed in
Appendix~\ref{app:C3}).
\item \textbf{repaired variant}: a hypothetical-repair ArchMap in
which the 3 charts of the triangle are replaced by a unified
BigDecimal scale-2 \code{HALF\_EVEN} convention.
\end{itemize}

The derived residuals of the pre-repair run (\emph{head}) show mismatching section values on the 3
triangle edges and the 2 preserve-family edges, and agreement on
consign--consign-price; the selected complex forms a single connected
component. The odd parity around the triangle is not solvable by
$\delta^0$, and the residual stands outside the boundary.

The results of the diagnosis staircase are:

\begin{center}
\small
\begin{tabular}{@{}p{0.30\linewidth}p{0.62\linewidth}@{}}
\toprule
Stage & Result \\
\midrule
head analyze & \code{MEASURED\_NONGLUING\_RESIDUAL}
(\code{run:78c31d6a3172}) \\
\quad grounding & \code{measured\_zero} --- each chart satisfies its
own local equations \\
\quad residual derivation & mismatch on the 3 triangle edges + 2
preserve edges; consign--consign-price agrees (all derived from
observation) \\
\quad boundary membership & \code{measured\_nonzero}
(\code{inB1: false}; with no triple declared, class vocabulary is not
unlocked --- made explicit by a named boundary statement) \\
gate head & \code{BLOCKED\_BY\_GATE\_POLICY} \\
repaired analyze & \code{REPAIR\_GLUES\_WITHIN\_SELECTED\_COMPLEX}
(\code{run:6685bab8db21}; the remaining preserve residual is inside
$B^1$) \\
compare head$\to$repaired &
\code{MEASURED\_OBSTRUCTION\_NO\_LONGER\_RECORDED\_AFTER\_CHANGE}. The
reading of the run pair is that the difference of the residuals of the
two runs is not solvable by $\delta^0$, consistent with the change
from head outside the boundary to repaired inside the boundary
(Appendix~\ref{app:C1}) \\
gate repaired & \code{PASS\_WITHIN\_GATE\_POLICY} \\
\bottomrule
\end{tabular}
\end{center}

The complexes and verdicts before and after repair are contrasted in
Figure~\ref{fig:onecent}.

\begin{figure}[htbp]
\centering
\includegraphics[width=\linewidth]{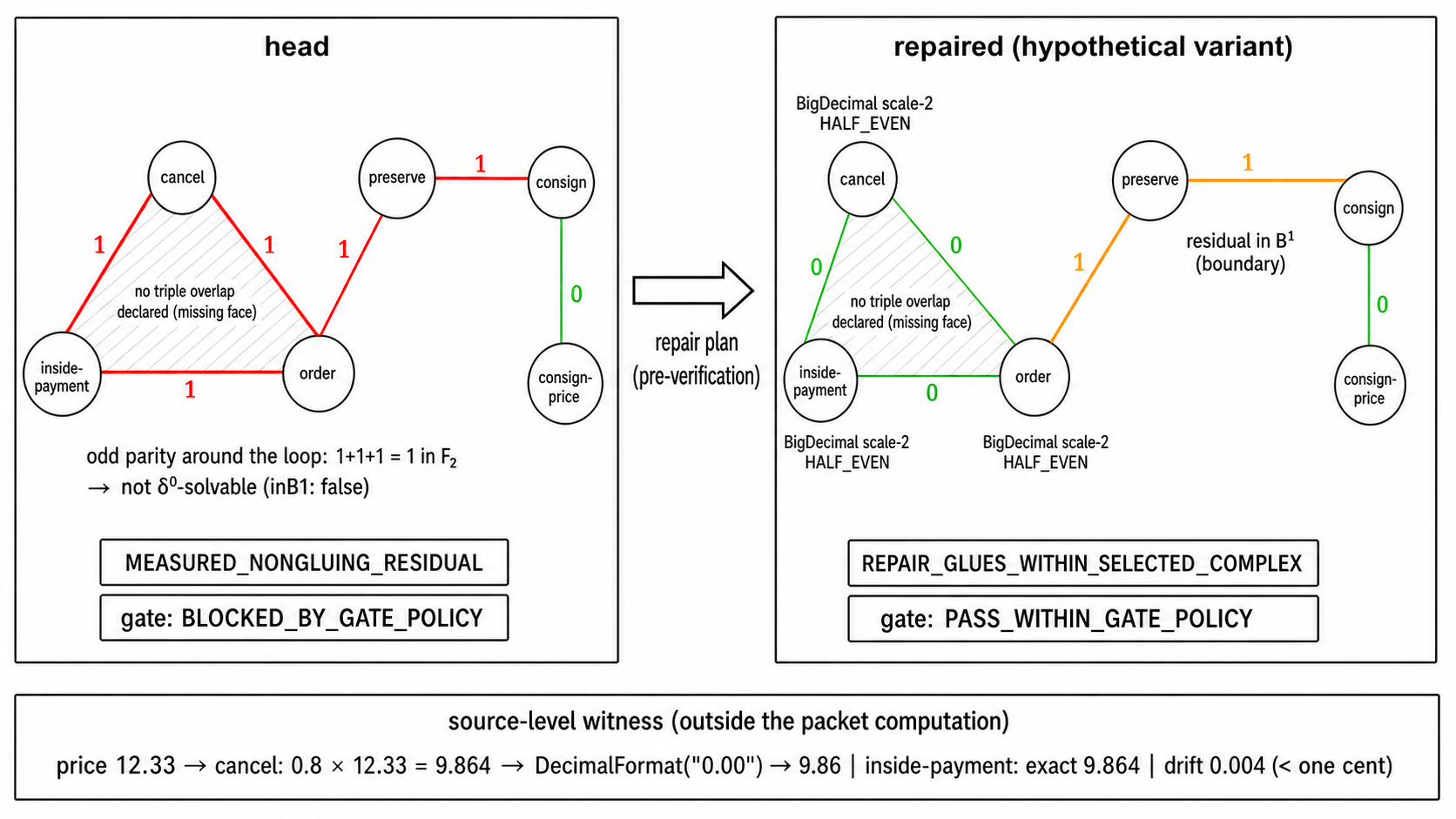}
\caption{The diagnosis staircase of the one-cent obstruction. In head
(left), derived residuals stand on the 3 triangle edges and the 2
preserve-family edges, and the odd parity on the closed loop of the
triangle is not solvable by $\delta^0$ (\code{inB1: false}). No triple
overlap is declared; the absence of the face is shown hatched. In the
repaired variant (right), replacing the 3 charts of the triangle by a
unified BigDecimal scale-2 \code{HALF\_EVEN} convention removes the
triangle residuals, and the remaining preserve-family residual falls
inside $B^1$. The numerical witness in the lower band is a
source-level recomputation by the operator and does not appear in the
packet computation (\S\ref{sec:71}).}
\label{fig:onecent}
\end{figure}

\subsection{What the measurement showed}\label{sec:73}

\paragraph{Reality of locally consistent, globally non-gluing.}
The \code{measured\_zero} of grounding says, as measurement, that each
chart satisfies its own local equations (in the sense of the
displayed-equation check of Appendix~\ref{app:C1}). Cancel is faithful to
cancel's rounding convention, inside-payment to exact arithmetic,
order to pass-through storage. Each pairwise handoff holds as well.
What was observed as measurement is that the odd parity of the
convention mismatch is not solvable by $\delta^0$ on the closed loop
(\code{inB1: false}). Under the assertion of the absence of a face
(triple) filling the loop (\S\ref{sec:71}), this measurement becomes
an instance --- not artificially planted, but found in real OSS --- of
the central SAGA structure: locally consistent, globally non-gluing.

\paragraph{Every step of the staircase works on derived residuals alone.}
From the measurement of non-boundary residuals derived from
observation (ArchMap) and law / equations (law surfaces,
MeasurementProfile), through blocking by the gate, prior verification
of the repair plan, the recording by compare of the disappearance of
the obstruction and of the residual difference of the run pair, to
gate PASS --- the full circle was walked without supplying any
authored residuals, certificates, or comparison data. The only operator
declarations that remain are the selected complex (\S\ref{sec:74}),
the witness bindings (law surface), and the repaired variant.

\paragraph{Effectiveness of the mathematical discipline.}
Declaring the drift-bearing triangle itself as a triple is rejected by
the cocycle condition --- a mathematically legitimate rejection.
Mismatches without witness bindings and references without grounds
were rejected fail-closed. The discipline held under the load of real data.

\paragraph{Executed silence.}
The measurement axis requiring runtime measured values (axis name
\code{harmonic-debt}) was not supplied in the absence of actual
measurements and was treated as silence. Likewise, not emitting class
vocabulary on a complex without triple declarations is the same
discipline in action; ArchSig recorded that boundary as a named
boundary statement, minimally, near the conclusion.

\subsection{Essentials of the input contract}\label{sec:74}

The inputs of ArchSig are the observation (ArchMap), the declarations
on the law / equation side (LawPolicy, law surfaces,
MeasurementProfile), and a repair plan that declares only the selected
complex. The per-artifact input contract and the definitions of the
computations --- grounding, residual derivation, boundary membership,
\code{compare}, \code{gate} --- are fixed in Appendix~\ref{app:C1}.

Residuals are not inputs. Under this contract, the measurement claim
of this paper stands on the following. \textbf{The conclusions of the
SAGA diagnosis --- residuals, boundary membership, gate verdicts ---
are derived deterministically from the Atom observation and the
selected equation system. What the operator can write are choices: the
scope of observation, the selected complex, the witness bindings; no
input carrying conclusions exists in this contract.} This separation
is enforced by fail-closed checks: charts outside the selected cover,
overlaps without observed restrictions, duplicated overlap
declarations, unobserved sections, and mismatches without witness
bindings all stop the computation without generating conclusions.
Therefore, while obstructions can be silenced by selection
(enumeration completeness is disclosed as an assumption in
Appendix~\ref{app:C3}), no nonzero residual can be erected on an edge
where the observations agree.

What this computation executes is a finite fragment of the complex
vocabulary of Section~\ref{sec:4}: residual derivation on the selected
1-skeleton, $B^1$ membership, and the residual difference of a run
pair reduce to finite \code{F2} linear algebra. The finite
instantiation of the comparison of Theorem~\ref{thm:central} ($\chi$,
$\Phi$, $\kappa$, presentation exactness) is not within the scope of
this computation; it is carried by the Lean witnesses of
Section~\ref{sec:6} (\S\ref{sec:75}).

On the supply side, providing an ArchMap involves semantic reading of
how the sources are used --- a probabilistic process that cannot be
replaced by a deterministic computation. This work does not hide that
stage: it fixes it as a procedure executed by an AI agent (the
authoring SKILL), chaining the probabilistic semantic reading and the
deterministic measurement in series. The discipline of the process and
the two layers of reproducibility are fixed in Appendix~\ref{app:C2}.
The kinds of conditions behind the conclusions (\code{computed} /
\code{checked} / \code{assumed} / \code{unmeasured}) are fixed row by
row by the condition matrix of Appendix~\ref{app:C3}, and the
reproduction procedure by Appendix~\ref{app:C4}.

\subsection{Boundary of the claims}\label{sec:75}

The claims of this case study are limited to the following scope. The
detection of the convention mismatch itself was carried by the stage
of the closed-equational surface. What the SAGA stage added is: the
descent reading of the same observation as boundary membership on the
selected 1-skeleton; the explicit limitation of the grounding reading
(\S\ref{sec:73}); the prior verification of the repair plan; the
reading of the residual difference of the run pair (Appendix~\ref{app:C1});
and the consistent diagnosis of the gate. We do not claim that ``SAGA
discovered a new failure''.

What is authored is selection. The selected complex (repair plan), the
witness bindings (law surface), and the repaired variant are
declarations written by the operator, and none of them carries
residual values. Class vocabulary is not unlocked: in this packet,
where no triple is declared on the component containing the triangle,
ArchSig's reading stays at boundary membership (Appendix~\ref{app:C3}). The
repaired variant is a hypothetical-repair ArchMap with rewritten
sections, and what \code{PASS\_WITHIN\_GATE\_POLICY} shows is the
mechanism of prior verification --- ``with this repair proposal, the
system would glue''. The frequency of the drift and its monetary
magnitude require runtime measurement and are not measured in this
paper. Broad benchmark evaluation and general detection performance
are matters for separate empirical studies.

The finite instantiation of Theorem~\ref{thm:central} is outside the
scope of this case study. Its bearer is the Lean witness of
Section~\ref{sec:6} (the circle witness of Example~\ref{ex:circle});
what the packet of this section carries are the finite checks shown by
the condition matrix of Appendix~\ref{app:C3}. No correspondence transporting
the measurement run into Lean is required.

We also make explicit the coverage of the SKILL-based supply process. The authoring SKILL of Appendix~\ref{app:C2} covers the
supply of the ArchMap, and an authoring SKILL of the same kind is in
place for the selected complex (repair plan). The law surfaces and the
repaired variant were supplied, at the time of this experiment, as
builder scripts and authored declarations. The builder scripts and
supply findings of this experiment are recorded as design material for
that step.

\section{Related Work}\label{sec:8}

We situate the main result of SAGA within four groups:
(1) cohomological program analysis; (2) sheaves for architecture and
systems engineering; (3) architecture conformance and formal
connection; (4) mechanized theorem chains and executable measurement.

\subsection{Cohomological program analysis}\label{sec:81}

Young unifies type checking, bug finding, and program equivalence as
\v{C}ech-cohomological analysis of semantic presheaves over observation
sites of Python programs \citep{young2026sheaf}. In its use of local
semantic observation and \v{C}ech $H^1$, and in reporting a Lean
formalization and an executable analyzer, it is the contemporaneous
work closest to SAGA.

The two share the following structure:
\begin{center}
\small
local semantic observations $\to$ site / cover $\to$ presheaf or
coefficient system $\to$ \v{C}ech complex $\to$ $H^1$ obstruction
$\to$ executable analysis
\end{center}

The differences lie in the mathematical center and the level of
abstraction. Young computes cohomology directly on the
program-semantic presheaf of Python. AAT starts from
language-independent Atoms and an equation system, relating facts
spanning different languages, services, and storage representations on
a single architecture object. SAGA constructs semantic-repair
cohomology and equation-generated AAT \v{C}ech cohomology
\emph{separately} and proves a comparison theorem identifying their
$H^1$. The difference is made concrete by the one-cent case: the
amounts in question are all represented as \code{string}, and the
implementation types agree. What SAGA measures is not a difference of
type names but whether locally valid semantic conventions glue into
one global architecture.

\begin{table}[htbp]
\centering
\scriptsize
\begin{tabular}{@{}p{0.20\textwidth}p{0.34\textwidth}p{0.38\textwidth}@{}}
\toprule
Axis & Young 2026 & SAGA \\
\midrule
primary object & Python programs and observation sites &
language-independent Atom families and architecture objects \\
coefficient & semantic presheaf, $\Ftwo$ realization & the
$\QE=\ObsE/\IOb$ generated by the equation system \\
central theorem & program-analysis claims via \v{C}ech cohomology &
the comparison $H^1_{\mathrm{sem}} \cong \check H^1(\cU,\QE)$ \\
repair reading & rank = independent fixes & residual classes,
boundary, repair gluing \\
empirical unit & 375 program-analysis benchmarks & a real microservice
architecture with before/after-repair comparison \\
domain-specific construction & direct computation on the
program-semantic presheaf & two independent presentations (semantic
repair / equation quotient), the construction of $\chi/\Phi/\kappa$,
exactness conditions relativized to architecture data \\
discharged obligation & connecting cohomology computations to
program-analysis claims & derivation of soundness, checking of the
$\ker$/$\operatorname{im}$ conditions, proof of the semantic
correspondence of residual classes \\
\bottomrule
\end{tabular}
\caption{Comparison of Young 2026 and SAGA.}
\label{tab:young}
\end{table}

What SAGA claims as the novelty of the comparison is the last two rows
of Table~\ref{tab:young} --- the domain-specific construction and the proof
obligations discharged there --- not the general mechanism by which a
coefficient isomorphism induces a cochain isomorphism (the attribution
of what is proper to SAGA is in \S\ref{sec:51}).

Young's Lean formalization (1{,}259 lines), the 375 benchmarks, and
the evaluation numbers are cited as results reported by that paper.
Since no public link to the code artifact of that paper's analyzer
could be confirmed, source-level comparison is outside the scope of
this paper.

In the lineage of global-section obstructions, we place
sheaf-theoretic contextuality \citep{abramsky2011sheaf} and its
non-vanishing obstructions in \v{C}ech cohomology
\citep{abramsky2012cohomology} as foundational references. As prior
examples of computable obstructions, we cite the application of
\v{C}ech cohomology to CSP and structure isomorphism
\citep{oconghaile2022cohomology}, and the task sheaf for distributed
task solvability \citep{felber2025sheaf}.

\subsection{Sheaves for architecture and systems engineering}\label{sec:82}

Gibson gives a Lean-verified sheaf model for multi-view consistency in
model-based systems engineering, characterizing global designs by
compatible local designs on pairwise interfaces
\citep{gibson2026sheaves}. SAGA develops the complementary
obstruction-theoretic direction: constructing the $H^1$ class that
remains when local data do not glue, and proving the comparison
between its semantic realization and its equation-generated
realization.

As the historical starting point we place the sheaf semantics of
objects and interaction \citep{goguen1992sheaf}. As foundational
references for cellular sheaves and finite computation we cite
\citet{curry2014sheaves}, \citet{robinson2017sheaves}, and
\citet{hansen2019spectral}.

\subsection{Architecture conformance and formal connection}\label{sec:83}

The architectural-mismatch analysis of \citet{garlan1995architectural}
examined the mismatch of assumptions embedded in components,
connectors, and the construction process. SAGA moves this kind of
mismatch into the mathematics of local equations and global gluing,
treating the residue as an $H^1$ class. Reflexion models \citep{murphy1995reflexion}, which
compute convergence / divergence between a high-level model and a
source model, are the representative starting point of conformance.
Wright \citep{allen1997formal} formalized the compatibility of
connector protocols in CSP. What SAGA treats is not protocol
conformance but the cohomology class of whether multiple local
semantic repairs glue across a cover. Compared with the survey of the
expressive power of ADLs \citep{medvidovic2000classification} and the
systematic surveys of architecture erosion
\citep{desilva2012controlling, li2022understanding}, SAGA offers a
global obstruction class for selected local data and a repair
comparison.

\subsection{Mechanized theorem chains and executable measurement}\label{sec:84}

Lean~4 \citep{demoura2021lean} and Mathlib \citep{mathlib2020lean} are
cited as the provenance of the formalization environment. The
formalization contribution of SAGA is not the use of Lean itself, but
the construction, as a single machine-checked theorem chain reaching
the conclusion bundle of the central theorem, of the semantic repair
presentation, the cover-relative \v{C}ech complex, the coefficient
isomorphism and cochain comparison, true sheaf descent, and the finite
witnesses of the zero and nonzero cases (Section~\ref{sec:6}). The
comparison with the formalizations of Young and Gibson is made in
terms of the scope of the theorem chain and the connection to
executable measurement.

\subsection{Synthesis}\label{sec:85}

Prior work has established sheaves as a language for global
consistency, cohomology as computable obstruction, and formal
architecture models as a basis for conformance analysis. To these
lineages SAGA adds the construction of an equation-generated AAT
\v{C}ech complex for software architecture and the proof that its
first cohomology agrees with an independently defined semantic-repair
obstruction. Lean verifies the comparison, and ArchSig evaluates its
finite architectural instances.

\section{Discussion and Research Outlook}\label{sec:9}

The SAGA comparison theorem is the first theorem about the
local-to-global capability of AAT to reach a full proof, and it is
also the first traversal of a larger research program. This section
discusses what the theorem has settled (\S\ref{sec:91}) and the
research directions it opens (\S\ref{sec:92}--\S\ref{sec:96}).
Everything from \S\ref{sec:92} onward is research outlook, not to be
confused with proved results; the development of theorems,
implementations, and tooling in each direction is outside the scope of
this paper.

\subsection{What SAGA has settled}\label{sec:91}

The first consequence of the comparison theorem is
\textbf{translatability}. A diagnosis posed in the language of
semantic repair and a computation posed in the language of equation
geometry agree at the level of $H^1$ and residual classes. Subsequent
theory development can proceed on either side and be translated to the
other: structure discovered on the semantic side can be computed with
geometric tools, and theorems proved on the geometric side can be read
in the language of repair.

The second consequence is the \textbf{grounding of counterfactuals}.
The sentence ``without this fix, the whole would not have glued'' used
to be a counterfactual with no place to stand within the vocabulary of
testing and observation, because a failure that did not occur cannot
be observed. The three-way equivalence of Theorem~\ref{thm:gluing}
gives this sentence a mathematical status. Since the zero class is
equivalent to the existence of a global repair, ``it glued because the
correction killed the class'' is a consequence, not a heuristic. The
before/after-repair comparison of Section~\ref{sec:7} is the first
working diagnostic example of how this consequence operates in an
engineering process (the division of labor between the scope of the
measurement checks and the Lean witnesses of Section~\ref{sec:6} is in
\S\ref{sec:75}).

\subsection{The Rising Sea as method}\label{sec:92}

The method of this research program follows what Grothendieck called
\emph{la mer qui monte} --- the rising sea \citep{grothendieck2022recoltes}.
Instead of prying open a hard problem with ad hoc tools, one raises
the water level of abstraction of the theory until the problem
dissolves naturally at the higher level.

The engineering problem ``locally correct but globally broken'' has,
at the level where one works directly with code and modules, remained
a target attacked with the local tools of testing, contracts, and CI.
Because the problem itself lives in the overlaps of parts, local
tools cannot reach it in principle. Rising to the level of Atoms,
sites, sheaves, covers, and cohomology, the same problem becomes an
ordinary object with coordinates: an end-to-end failure is a class in
$H^1$; classes can be computed; and what can be computed becomes the
work of tools.

This raising of the water level has two distinctive features. First,
\textbf{abstraction is portability}. Because Atoms are placed at a
level independent of language and framework, neither the theory nor
the measurement depends on a particular technology stack, and a
multi-language system can be treated as a single architecture
geometry. Second, \textbf{the water level is maintained by Lean}.
Grothendieck's sea settled in thousands of pages of prose; the sea of
this program is maintained as a tower of machine-verified theorems,
and as long as the theorems are formally verified, the water does not
recede. That the name SAGA is an homage to Serre's GAGA
\citep{serre1956gaga} --- the comparison theorem that established the
correspondence of two geometries --- reflects this configuration.

\subsection{The next mathematical peaks}\label{sec:93}

From the SAGA mathematics proved in this paper, the following theory
developments can be formulated naturally.
\begin{itemize}
\item \textbf{Characterization of the descent condition}: the
three-way equivalence of Theorem~\ref{thm:gluing} holds under the
selected true sheaf condition. The next peak is to characterize the
validity of this condition itself from architecture data and lower it
to checkable assumptions. Once realized, ``this architecture satisfies
the descent assumptions, so passing local verification implies global
consistency'' becomes a checkable measurement item, and the theory can
point to the places where the assumptions break as exactly the places
where integration verification should concentrate.
\item \textbf{Higher coherence}: developing as independent statements
what this paper deliberately separated (\S\ref{sec:57}): the
pointed-set-valued $H^1$ of nonabelian torsors, 2-cocycles and gerbes,
obstructions in $H^2$ and beyond, stack descent.
\item \textbf{Architecture schemes and morphism theory}: treating the
zero locus of the required equations as a scheme, and proving theorems
about morphisms, base change, and fibers between architectures.
\item \textbf{Moduli of repairs and derived deformation}: treating the
space of repairs itself as a geometric object, speaking of
deformations, degenerations, and singularities of repairs.
\item \textbf{First-order conormal specialization}: the specialization
of the conormal geometry given by the first-order approximation of the
obstruction ideal.
\item \textbf{Deepening the Atom foundations}: refinement of the Atom
axiom system and extension of the correspondence theory between
semantic atoms and equation Atoms.
\end{itemize}

\subsection{From statics to dynamics: Software Field Theory}\label{sec:94}

AAT is a theory of statics: it treats the consistency of an architecture at a
point in time. Software Field Theory (SFT) is the research program
that constructs the dynamics of software evolution on top of that
geometry: organizing development traces as a site, treating changes,
branches, and merges in geometric language, and formulating review,
CI, and operational feedback as forces acting on the development
system.

The central viewpoint of SFT is a change in the role of architecture:
an architecture is not only the present shape of the code but the
\textbf{shape of the reachable futures}. A good architecture is a
field configuration in which desirable futures are reachable and
undesirable futures are hard to reach. Under this view, SFT aims at
central propositions of the following form:

\begin{quote}
A software architecture is modular exactly when its future evolution
satisfies descent.
\end{quote}

A merge in version control is a gluing of local changes, and SFT reads
it as descent theory. The static comparison proved by SAGA is the
ground this dynamics stands on at each instant. Where the static
$H^1$ measures ``does not glue now'', SFT aims to measure ``can this
sequence of changes reach a future that glues''.

\subsection{The future of measurement: a fixed point in the age of generation}\label{sec:95}

As research outlook, we also record one configuration on the
measurement side.

In development where code is generated faster than humans can
understand it, the code, the tests, the reviews, and the fix proposals
all sway on the same probabilistic generative ground. Within that, a
deterministic measurement grounded in theorems becomes a point of a
different nature: its verdict is not someone's opinion but the
consequence, on a selected contract, of a machine-verified theorem.

In this configuration, ArchSig's design principle of ``staying silent
about what cannot be spoken'' works as a safety device. Agents read
outputs literally and act on them literally, so a checker that warns
without grounds makes the generative loop diverge. Only a checker that
speaks exactly what can be spoken from the given contract can be
placed as a gate of a generative system. Economically as well, the
computational cost of deterministic verification is small compared to
probabilistic generation, and the faster generation becomes, the
higher the value of computation that rejects mistakes before
execution.

Human work does not disappear in this configuration; it moves. Which
vocabulary is admitted, which equations are required, what counts as
PASS --- the selection of the LawPolicy remains on the human side as
the decision of what deserves protection --- indeed, it is refined into
precisely that. The diagnosis staircase of Section~\ref{sec:7} (nonzero residual
$\to$ BLOCKED $\to$ prior verification of the repair $\to$ PASS) is
the minimal working example of this configuration.

\subsection{The Rising Sea research program}\label{sec:96}

The research program as a whole is the following chain:
\begin{center}
\small
put architectural facts into words
$\to$ construct AAT geometry from Atoms and equation systems
$\to$ formalize the mathematical claims in Lean
$\to$ measure selected finite instances
$\to$ SFT connects the geometry to evolution dynamics
$\to$ compute trajectories, reachable futures, feedback, governance
\end{center}

SAGA is the first instance traversing the mathematics, formalization,
and measurement stages of this chain in a single theorem. The provenance discipline this paper has fixed (classifying
the kinds of claims and connecting each claim to primary evidence)
will be used in the same form on all subsequent peaks. When the water
level of the theory has risen high enough, users see only the surface.
An engineer, having learned neither categories nor sheaves nor
cohomology, says ``analyze this'', receives only the word
\emph{obstruction}, fixes it, and moves on. That quietness is the goal
of this program.

\section{Conclusion}\label{sec:10}

This paper has presented three results.

First, the \textbf{proof of the SAGA comparison theorem}. From the
semantic repair presentation we constructed $\Msem$ and the semantic
\v{C}ech complex, and from the equation system $\QE$ and the geometric
\v{C}ech complex, independently; under SAGA presentation exactness we
proved $\Phi:\Msem\simeq \QE$, the cochain isomorphism $\kappa$, the
$H^1$ isomorphism, and the residual class correspondence
$\kappa_*([\rsem])=[\rE]$. Under the true sheaf condition, the zero
class is equivalent to the existence of an actual global repair, and a
global section is constructed from the correction via the corrected
matching family and sheaf amalgamation.

Second, the \textbf{state of the Lean formalization}. The conclusion
bundle of the central theorem (residual correspondence, zero/nonzero
equivalence, grounded gluing), each stage of the coefficient
isomorphism and the cochain comparison, and the finite witnesses of
the zero and nonzero cases are recorded as a machine-checked theorem
chain, with the correspondence to the theorems of this paper, the
assumptions, and the axiom status fixed at declaration granularity.

Third, the \textbf{one-cent realization}. On the refund triangle of a
real microservice system, the non-boundary residual erected by the
collision of three monetary conventions was derived and measured from
observation (on the selected complex, Appendix~\ref{app:C3}), and the full
circle --- blocking by the gate, prior verification of the repair
proposal, disappearance of the obstruction after repair, gate PASS ---
was walked as one reproducible computation. Each chart satisfied its
own local equations. The failure appeared only on traversing the loop,
as a drift below one cent belonging to no local part of the system.

The voyage that started from an abstract comparison theorem reached,
through machine verification, one cent in real code. This one cent is
the smallest concrete witness that the sum of local correctness does
not reach global correctness. For SAGA this landfall is a Cape of Good
Hope: the first cape where the theory breaks out into the real sea,
and the routes beyond --- architecture schemes, moduli, Software Field
Theory, the Rising Sea --- are opened by this comparison theorem.

\section*{Acknowledgments}

The construction of the theory, the Lean formalization, the tooling
implementation, and the writing of this paper were carried out under the
author's direction in collaboration with LLM agents (Claude, Codex). The
correctness of the formalized mathematical claims is fixed by machine
verification in Lean, and the probabilistic process of observation supply is fixed as
artifacts by the discipline of Appendix~\ref{app:C2}. The AI agents are not
authors of this paper; the author bears responsibility for all conclusions.

\bibliographystyle{plainnat}
\bibliography{zenodo_saga_references}

\appendix

\section{Attribution of Assumptions (Which Proof Consumes What)}\label{app:A}

This appendix tabulates where the proofs in the text consume inputs
1--8 of Theorem~\ref{thm:central}. The content of the claims is
governed by the corresponding sections of the text; this appendix is an
index of attribution.

\subsection{Comparison core (\S\ref{sec:53}--\S\ref{sec:55})}\label{app:A1}

\begin{itemize}
\item Consumed: the cover intersection diagram; the two coefficient
presheaves (assumptions 1--2); the generator map and the two
completeness conditions (assumptions 3--4); the local-state data
(assumptions 5--7, used in the derivation of soundness).
\item Not used: the global sheaf condition, the finite enumeration of
the cover, displayed equation fulfillment.
\end{itemize}

\subsection{Actual gluing (\S\ref{sec:56})}\label{app:A2}

\begin{itemize}
\item What Lemma~5.2A consumes is the monomorphism property of the
cover and, of assumption~8, the subsingleton clause for $\Psem(V)$
(the proof consumes only the pairwise-overlap part). The lemma is used
in the completion of the corrected family in the forward proof and in
the gluing step of the equation-side globalization.
\item In addition, actual gluing uses the true sheaf condition of
(iii) (the sheaf condition and membership in the topology). The
$\Psem(V)$ subsingleton clause is supplied as the same proposition as
condition~(4) of the true sheaf condition.
\end{itemize}

\subsection{Unconsumed clauses}\label{app:A3}

Of assumption~8, the subsingleton clause for $\PE(V)$ and the
coefficient-vanishing clause $\Msem(V)=\QE(V)=0$ are consumed by no
proof in this paper. The complexes run over nonempty intersections
only (\S\ref{sec:41}), and what the equation-side globalization (end
of \S\ref{sec:56}) uses for empty overlaps is also the
$\Psem$-side subsingleton clause (Lemma~5.2A). These unconsumed
clauses are retained in order to fix the input side (assumption~8) in
a form that includes the full identification with the ordered
\v{C}ech complex of \S\ref{sec:41} (outside the scope of this paper).

\subsection{Consumption in the Lean formalization (thin sites)}\label{app:A4}

The Lean formalization is carried out over a site whose context
category is thin (\S\ref{sec:57}). By thinness --- the agreement of
parallel morphisms --- every morphism in this model is automatically a
monomorphism, and the treatment of self-overlaps and reversed overlaps
in Lemma~5.2A is also settled automatically by the agreement of
parallel morphisms. Consequently the formalized counterpart of
Lemma~5.2A, \code{SiteStateData.matchingFamily\_iff}, consumes the
selected pairwise overlaps and their lift data, but does not consume
the monomorphism property as an explicit hypothesis (the commutativity and uniqueness
of pullbacks follow from thinness). The gluing argument over general
$\ArchCtx(X)$ is \code{unported} (Section~\ref{sec:6}).

\section{Lean Source Correspondence (Declarations, Labels, Source Files)}\label{app:B}

This appendix fixes the organization, within the release tag, of the
Lean sources that are the primary evidence for the status of
Section~\ref{sec:6}. The SAGA theorem chain is located under
\code{Formal/AG/SemanticRepair/Saga/}. The kernel axiom audit is
implemented by
\code{\#assert\_}\allowbreak\code{standard\_}\allowbreak\code{axioms\_only}
in \code{Formal/AG/}\allowbreak\code{AxiomAudit.lean} (an allowlist
check of all registered declarations), which the release CI executes as
\code{lake env lean} \code{Formal/AG/}\allowbreak\code{AxiomAudit.lean}.

The docstrings of the Lean sources use labels from a numbering series
distinct from this paper, originating in internal organization during
development (normalized here as \code{X.Theorem 1.1} and so on). The theorem
correspondence with this paper is given completely by the status table
of Section~\ref{sec:6}; the labels serve for cross-checking when
browsing the Lean sources. The correspondence with the declarations of
Section~\ref{sec:6} is given in Table~\ref{tab:leandecl} (source files
are relative to \code{Saga/}).

\begin{table}[htbp]
\centering
\scriptsize
\begin{tabular}{@{}p{0.46\textwidth}p{0.22\textwidth}p{0.24\textwidth}@{}}
\toprule
Lean declaration & Lean-side label & Source file \\
\midrule
\code{SagaEquationPacket.sagaCentralTheorem} & X.Theorem 1.1 &
\code{TrueSheafDescent.lean} \\
\addlinespace
\code{PrimaryStateCorrespondence.}\allowbreak\code{relationSound\_of\_stateCorrespondence} &
X.Lemma 6.2A & \code{Exactness.lean} \\
\addlinespace
\code{SagaEquationPacket.phiEquiv} (generic version:
\code{PrimaryCoefficientCorrespondence.phiEquiv}, X.Theorem 6.3 /
Corollary 6.7) & as listed & \code{EquationRealization.lean},
\code{Exactness.lean} \\
\addlinespace
\code{ExactnessFixtures.soundness\_failure} /
\code{completeness\_failure} / \code{generation\_failure} &
X.Example 6.6 & \code{Exactness.lean} \\
\addlinespace
\code{kappa1\_delta0}, \code{kappa2\_delta1} & X.Theorem 7.2 &
\code{KappaComparison.lean} \\
\addlinespace
\code{SagaEquationPacket.kappaStarAddEquiv} (generic version:
\code{kappaH1AddEquiv}, X.Theorem 7.4) & as listed &
\code{KappaComparison.lean} \\
\addlinespace
\code{SagaEquationPacket.residual\_correspondence\_class},
\code{betaAligned\_residual} & X.Theorem 7.5 &
\code{KappaComparison.lean} \\
\addlinespace
\code{SagaEquationPacket.globalRepair\_nonempty\_iff},
\code{sagaGroundedGluing} & X.Theorem 8.2 &
\code{TrueSheafDescent.lean} \\
\addlinespace
\code{SiteStateData.matchingFamily\_iff} & X.Lemma 2.1A &
\code{OrderedComparison.lean} \\
\addlinespace
\code{CircleWitness.semanticResidualClass\_ne\_zero},
\code{circle\_nonzero\_class\_transfer} & X.Example 10.2 / Appendix
B.9 / X.Theorem 7.6 & \code{CircleWitness.lean} \\
\addlinespace
\code{DescentWitness.descentTrueSheaf},
\code{descent\_sagaGroundedGluing} & X.Definition 8.1 / X.Theorem 8.2 &
\code{DescentWitness.lean} \\
\bottomrule
\end{tabular}
\caption{Correspondence between declarations, Lean-side labels, and
source files.}
\label{tab:leandecl}
\end{table}

Table~\ref{tab:leanfiles} gives the correspondence between the building
blocks of \S\ref{sec:61} and the source files.

\begin{table}[htbp]
\centering
\scriptsize
\begin{tabular}{@{}p{0.55\textwidth}p{0.36\textwidth}@{}}
\toprule
Building block (\S\ref{sec:61}) & Source file \\
\midrule
cover, intersection diagram, three-term \v{C}ech complex &
\code{Cover.lean}, \code{CechThreeTerm.lean} \\
semantic repair presentation and $\Msem$ & \code{Presentation.lean} \\
affine repair system, selected atlases, residuals &
\code{RepairTorsor.lean}, \code{EquationLift.lean} \\
soundness derivation, exactness, coefficient isomorphism, finite
counterexamples & \code{Exactness.lean} \\
equation-side realization / production &
\code{EquationRealization.lean}, \code{EquationProduction.lean} \\
cochain comparison $\kappa$, $H^1$ isomorphism, residual
correspondence & \code{KappaComparison.lean} \\
true sheaf descent, final bundling of the central theorem &
\code{TrueSheafDescent.lean} \\
finite witnesses of the zero and nonzero cases &
\code{CircleWitness.lean}, \code{DescentWitness.lean} \\
poset site realization, ordered \v{C}ech model comparison bridge &
\code{OrderedComparison.lean}, \code{PartIVBridge.lean} \\
\bottomrule
\end{tabular}
\caption{Correspondence between the building blocks of
\S\ref{sec:61} and source files.}
\label{tab:leanfiles}
\end{table}

\section{Measurement Contract and Audit Surface of the ArchSig Case Study}\label{app:C}

This appendix fixes the input contract, the definition of the
computation, the supply process for observation, the kinds of
conditions behind the conclusions, and the reproduction procedure on
which the diagnosis of Section~\ref{sec:7} stands. The boundary of the
claims is stated in \S\ref{sec:75}.

\subsection{Input contract and computation}\label{app:C1}

This subsection fixes the input contract and the computation on
which the diagnosis of Section~\ref{sec:7} stands. The inputs of ArchSig are the following
artifacts.
\begin{itemize}
\item \textbf{ArchMap}: the Atom observation of the target system.
Records charts (local contexts), section values, and the
correspondence of overlaps.
\item \textbf{LawPolicy} (artifact name): fixes the Atom vocabulary in
use and the selected equation reading.
\item \textbf{law surface} (artifact name): the family of equation
surfaces used by the diagnosis (closed-equational, SAGA-grounded,
descent). The descent surface includes the declarations of witness
variables bound to mismatch edges.
\item \textbf{MeasurementProfile}: fixes the conditions of
measurement, including the coefficient (in this case study,
\code{F2}).
\item \textbf{repair plan}: declares only the selected complex
(charts, overlaps, optional triple overlaps, enumeration assertion)
and a reference to the target cover. It carries no residuals, no
coefficients, no comparison data.
\end{itemize}

The charts and overlaps of the repair plan are references into the
observation; declarations that do not resolve to an observed cover and
observed restrictions are not accepted. The remaining enumeration
assertion and the presence or absence of triple-overlap declarations
are author assertions: they accompany the conclusions as disclosure
rows of the assumption ledger, not as material from which conclusions
are derived (\S\ref{app:C3}). Thus the only thing the repair plan adds
of its own is disclosed assumptions.

ArchSig's \code{analyze} computes the following and emits them as a
measurement packet.
\begin{itemize}
\item Per-chart \textbf{grounding}: whether each chart passes the
displayed-equation check. The check targets the defect coordinates
selected by the law surface (the residuals of displayed equations, the
finite realization of the $\epsilon$ of \S\ref{sec:34}) and the
declared criteria --- not ``complete fulfillment of all equations on
the chart''.
\item Per-overlap \textbf{residual derivation}: derives an \code{F2}
value from the comparison of the section-value sets observed by the
two end charts of each edge of the selected cover, recording per edge
the provenance of the law surface's witness binding and the observed
atom references. A witness binding is a per-edge declaration
connecting a mismatch to a law-side violation coordinate (the finite
realization of the $\nu$ of \S\ref{sec:34}); it carries no instance
values. A mismatch without a binding has no law-side coordinate on
which the derived residual could stand, so it falls fail-closed to
non-computability.
\item \textbf{Boundary membership} on the selected 1-skeleton: the
finite \code{F2} computation of whether the derived residual belongs
to the $\delta^0$-image ($B^1$). The vocabulary of residual classes is
unlocked only when a triple overlap is declared on the connected
component where the residual stands and the cocycle check actually
runs. On components without triple declarations the selected $C^2$ is
zero and the cocycle condition holds automatically, so ArchSig emits
no class vocabulary and makes that boundary explicit in the packet as
a named boundary statement.
\end{itemize}

\code{compare} confronts the packets of two runs and records the
change in the obstruction together with the reading of the run pair.
The reading of the run pair is the finite \code{F2} decision of
whether the difference of the derived residuals of the two runs
belongs to $\operatorname{im}\delta^0$ on the selected $C^1$.
\code{gate} returns the verdict \code{PASS\_WITHIN\_GATE\_POLICY} or
\code{BLOCKED\_BY\_GATE\_POLICY} from a packet and a gate policy.

We fix the correspondence between the three law surfaces and the
outputs. The closed-equational surface declares the per-edge equations
of convention agreement and carries the detection of mismatches (the
\code{cech} stage). The SAGA-grounded surface declares the per-chart defect
coordinates and criteria and carries grounding. The descent surface
declares the witness bindings on mismatch edges and carries residual
derivation and boundary membership (the \code{saga-descent} stage). The
``stages'' mentioned by the condition matrix of \S\ref{app:C3} and its
note~1 refer to this correspondence.

\subsection{Supplying the inputs: authoring of observation and reproducibility}\label{app:C2}

Supplying an ArchMap involves semantic reading of how the sources are
used. This reading is a probabilistic process and cannot be replaced
by a deterministic computation. This work does not hide that stage; it
fixes it as a process. That is, the creation of the ArchMap is defined
as a fixed procedure executed by an AI agent --- hereafter the
\textbf{authoring SKILL} --- turning the probabilistic reading into a
traceable artifact.

The SKILL fixes the provenance of observation by the following
discipline.
\begin{itemize}
\item \textbf{Separation of the mechanical layer and the reading
layer}: the mechanical layer performs only file enumeration, content
hashing, literal comparison of normalized keys, and referential
integrity checks. Generation of Atoms, semantic choices, adoption or
rejection of candidates, and similarity-based merging are not
permitted in the mechanical layer; they remain as recorded judgments
of the reading layer.
\item \textbf{Approval record of scope}: the revision of the target
repository, the include / exclude globs, and the approved scope
manifest are fixed as artifacts. Observation claims are bounded by the
recorded revision and the selected scope; extraction of all evidence
is not claimed.
\item \textbf{Independent two-pass reading and adjudication}: the
standard mode (full-dual) reads the same worklist twice independently,
takes an extraction diff, and adjudicates discrepancies by re-reading
the sources. Adoption decisions are recorded with counts and grounds.
\item \textbf{Authoring audit}: the integrated ArchMap becomes a
measurement input only after passing mechanical checks (referential
integrity, coverage ledger, and so on).
\item \textbf{Run records and explicit models}: each authoring run
records the model used. The ArchMap of all 42 services underlying this
case study (2{,}118 atoms / 43 contexts / 440 sources) was created by
a run whose extraction and adjudication subagents were pinned to a
lightweight model. The fact that observation supply does not demand an
expensive frontier model has thus been measured as a practical
condition of this measurement system.
\end{itemize}

Under this process design, reproducibility splits into two layers.
ArchSig's computation is deterministic: the same input artifacts
return the same outputs (verified by digests). The generation of the
ArchMap is probabilistic, but since scope, grounds, procedure,
adjudication, and audit are fixed in artifacts, a third party can
re-execute the observation under the same discipline and confront the
results. Chaining the probabilistic semantic reading and the
deterministic measurement in series is the design principle of this
supply pipeline.

\subsection{Conditions and input kinds: the condition matrix}\label{app:C3}

Table~\ref{tab:conditionmatrix} fixes, for each conclusion of the
diagnosis staircase, under which kind of condition it stands. \code{computed}
denotes a finite computation from the inputs; \code{checked} a
condition checked by ArchSig against finite artifacts; \code{assumed}
a premise declared by the author or the profile and recorded by the
packet in the assumption ledger; \code{unmeasured} an axis not
supplied and kept silent. The head and repaired packets record the
same distinctions.

\begin{table}[htbp]
\centering
\scriptsize
\begin{tabular}{@{}p{0.20\textwidth}p{0.17\textwidth}p{0.13\textwidth}p{0.42\textwidth}@{}}
\toprule
Condition & Kind & Status & Record \\
\midrule
finite cover & property of finite artifacts & \code{checked} & site
cover digest computed and recorded from normalized contexts, covers,
and the derived nerve \\
\addlinespace
residual derivation & comparison of observed sections &
\code{computed} & per-edge \code{F2} values, witness bindings, and
provenance of observed atom references; in head, the 3 triangle edges
+ 2 preserve edges mismatch \\
\addlinespace
witness bindings & law surface declaration & \code{checked} & mismatch
edges require witness variable bindings; unbound mismatches fall
fail-closed to non-computability \\
\addlinespace
coefficient (\code{F2}) & law-side selection & \code{checked} &
declaration of the selected MeasurementProfile; the repair plan
carries no coefficients \\
\addlinespace
enumeration completeness of the selected complex & author assertion &
\code{assumed} & the repair plan's enumeration assertion recorded as
an assumption ledger row \\
\addlinespace
triple absence / class vocabulary & author assertion & \code{assumed}
(class vocabulary not unlocked) & the selected complex declares no
triples (\S\ref{sec:71}); the reading stays at 1-skeleton boundary
membership, and a named boundary statement makes the boundary explicit \\
\addlinespace
boundary membership & finite \code{F2} computation & \code{computed} &
head: \code{inB1: false}; repaired: \code{inB1: true} \\
\addlinespace
U-adequacy, Leray-type comparison & profile-supplied premises &
\code{assumed} & U-adequacy is the premise that the selected cover
suffices for the target reading; Leray-type comparison is the premise
needed to compare the cover-relative reading with cover-independent
sheaf cohomology. Both are disclosed as ledger rows, and the latter
comparison is not claimed (\S\ref{sec:57}) \\
\addlinespace
torsor property, fixedness of the action, coefficient descent &
profile-supplied premises & \code{assumed} & the three premises for
reading local sections as torsors of the coefficient (torsor property
of local sections, fixedness of the action, descent of the
coefficient); 3 ledger rows \\
\addlinespace
restriction surjectivity & profile-supplied premise & \code{assumed} &
the premise that restriction maps are surjective onto overlaps; ledger
row \\
\addlinespace
forest nerve & profile-supplied premise & \code{assumed} (not
satisfied in this packet) & the ledger records the forest premise, and
the nerve computation of the same packet shows 1 cycle as
\code{computed} (note~1) \\
\addlinespace
quotient sheaf condition & law surface declaration & \code{assumed} &
the law-surface-side declaration that the coefficient system playing
the role of the quotient coefficient satisfies the sheaf condition;
disclosed as a ledger row \\
\addlinespace
residual difference of the run pair (head$\leftrightarrow$repaired) &
derived reading on the run pair & \code{computed} & $\delta^0$
solvability of the difference of the derived residuals of the two runs
(\S\ref{app:C1}); for this pair the difference is not solvable by
$\delta^0$ \\
\addlinespace
repaired ArchMap & hypothetical repair input & supplied /
hypothetical & a hypothesis variant expressing the unified convention;
\code{PASS\_WITHIN\_GATE\_POLICY} does not indicate an implemented
repair \\
\addlinespace
runtime monetary magnitude & empirical measurement & \code{unmeasured}
& \code{harmonic-debt} not supplied; silence. Frequency and amounts
are not included in the conclusions \\
\bottomrule
\end{tabular}
\caption{Condition matrix: each conclusion of the diagnosis staircase
and the kind of condition under which it stands.}
\label{tab:conditionmatrix}
\end{table}

Note 1: the forest nerve is a disclosed unsatisfied premise. The
nonzero reading of the \code{saga-descent} stage in head (the $B^1$
membership decision) does not depend on this row. The verdict of the
\code{cech} stage separately declares its dependence on this assumption.

\subsection{Reproduction}\label{app:C4}

All input artifacts, primary outputs, builder scripts, and the
authoring SKILL itself are included in the reproduction bundle. The
schema versions of the artifacts (for the repair plan,
\code{archsig-repair-plan/v0.5.7}) are fixed by the bundle manifest.
The \code{inputDigests} of the primary outputs have been verified to
agree with the canonical digests. Reproduction consists of executing
\code{analyze} (head / repaired), \code{compare}, and \code{gate}
$\times2$ from the fixed ArchSig version and inputs, and confirming
agreement of the runIds and the gate verdicts.

\paragraph{Release identity.}
Release tag \code{saga-paper-v1.0.0}; version DOI of the v1.0.0 deposit
\href{https://doi.org/10.5281/zenodo.21605207}{10.5281/zenodo.21605207}
(concept DOI \href{https://doi.org/10.5281/zenodo.21603761}{10.5281/zenodo.21603761},
which resolves to the latest version); ArchSig tool version \code{0.5.4}; artifact schema versions
\code{archsig-repair-plan/v0.5.7} and
\code{archsig-run-manifest/v0.5.4}. The tagged commit hash and the
reference to the release CI run are recorded in \code{MANIFEST.json}.
This text is version 1.0.1 (textual corrections to version 1.0.0; the
release identity is unchanged).

\paragraph{Deposit layout.}
The deposit bundle \code{saga-zenodo-bundle/} contains:
\code{paper/} (this PDF, the \LaTeX{} / BibTeX sources, and the two
figures); \code{evidence/} (the input artifacts and the primary
outputs of the case study: ArchMap head / repaired variants, law
policy, law surfaces, measurement profile, repair plans, gate policy,
the builder script, and the \code{analyze} $\times2$ / \code{compare}
/ \code{gate} $\times2$ outputs); \code{report/} (the canonical
diagnosis report, including the condition-matrix source);
\code{reproduction/} (the execution commands, the expected outputs,
and the authoring SKILL); \code{audit/} (the claim-to-evidence
matrix); \code{MANIFEST.json} (SHA-256 checksums of every bundle file
together with the release identity); and \code{CITATION.md}.

\paragraph{Reproduction.}
Obtain the repository at the release tag
(\code{git clone --branch saga-}\allowbreak\code{paper-v1.0.0}) and run, from the
repository root with \code{\$EV} pointing at the bundle's
\code{evidence/} directory (byte-identical to the repository's own
copy, as certified by the checksums):
\code{analyze} on the head inputs, \code{analyze} on the repaired
inputs, \code{compare} on the two run directories, and \code{gate} on
each measurement packet with the gate policy --- five invocations of
the ArchSig command-line tool, listed verbatim in
\code{reproduction/README.md}. The expected outputs are: head
\code{analyze} concludes \code{MEASURED\_NONGLUING\_RESIDUAL} with
runId \code{run:78c31d6a3172}; repaired \code{analyze} concludes
\code{REPAIR\_GLUES\_WITHIN\_SELECTED\_COMPLEX} with runId
\code{run:6685bab8db21}; \code{compare} records
\code{MEASURED\_OBSTRUCTION\_NO\_LONGER\_RECORDED\_AFTER\_CHANGE};
\code{gate} returns \code{BLOCKED\_BY\_GATE\_POLICY} for head (with a
nonzero exit code) and \code{PASS\_WITHIN\_GATE\_POLICY} for
repaired. The \code{inputDigests} of the primary outputs use
location-independent stable references, so agreement is checked by
runId and verdict rather than by output directory.

\end{document}